\documentclass[journal]{IEEEtran}
\usepackage{amsmath,amsfonts,amsthm,amssymb}
\usepackage{algorithm}
\usepackage{algorithmicx}
\usepackage{algpseudocode}
\usepackage{url}
\usepackage{subfigure}
\usepackage{multirow}
\usepackage{caption}
\usepackage{graphicx}
\usepackage[colorlinks=true, citecolor=blue, linkcolor=blue, urlcolor=blue]{hyperref}
\usepackage{cite}
\usepackage{orcidlink}
\usepackage{booktabs}
\usepackage{xcolor}

\newcommand{\R}{\mathbb{R}}
\newcommand{\C}{\mathbb{C}}
\newcommand{\st}{\;\text{s.t.}\;}
\newcommand{\sign}{\operatorname{sign}}
\newcommand{\T}{\mathrm{T}}
\renewcommand{\H}{{\mathrm{H}}}

\newcommand{\bc}{\mathbf{c}}
\newcommand{\bx}{\mathbf{x}}
\newcommand{\by}{\mathbf{y}}
\newcommand{\ba}{\boldsymbol{a}}
\newcommand{\bh}{\mathbf{h}}
\newcommand{\bq}{\mathbf{q}}

\newcommand{\bu}{\mathbf{u}}
\newcommand{\bv}{\mathbf{v}}
\newcommand{\bw}{\mathbf{w}}
\newcommand{\bz}{\mathbf{z}}
\newcommand{\bX}{\mathbf{X}}
\newcommand{\bA}{\mathbf{A}}
\newcommand{\bC}{\mathbf{C}}
\newcommand{\bD}{\mathbf{D}}
\newcommand{\bH}{\mathbf{H}}

\DeclareMathOperator*{\argmin}{arg\,min}

\newtheorem{theorem}{Theorem}
\newtheorem{lemma}[theorem]{Lemma}

\newtheorem{proposition}[theorem]{Proposition}

\begin{document}

\title{Cram\'er--Rao Bound Optimization for Massive MIMO \\ DFRC Systems with 1-Bit DACs and ADCs}

\author{Chenfei Huang,~\IEEEmembership{Graduate Student Member,~IEEE},
Mingjie Shao,~\IEEEmembership{Member,~IEEE},
\\ and Ya-Feng Liu,~\IEEEmembership{Senior Member,~IEEE}

\thanks{Chenfei Huang and Mingjie Shao are with the State Key Laboratory of Mathematical Sciences, Academy of Mathematics and Systems Science, Chinese Academy of Sciences, Beijing 100190, China (email: huangchenfei@lsec.cc.ac.cn; mingjieshao@amss.ac.cn).}
\thanks{Ya-Feng Liu is with the Ministry of Education Key Laboratory of Mathematics and Information Networks, School of Mathematical Sciences, Beijing University of Posts and Telecommunications, Beijing 102206, China (e-mail: yafengliu@bupt.edu.cn).}
}

\maketitle

\begin{abstract}

In this paper, we investigate the dual-function radar-communication (DFRC) design for massive multiple-input multiple-output (MIMO) systems equipped with 1-bit digital-to-analog converters (DACs) at the transmitter and 1-bit analog-to-digital converters (ADCs) at the receiver, motivated by the need for low-cost and power-efficient implementations of massive MIMO systems. 
We consider a downlink scenario where the transmit signal matrix is optimized to enhance sensing performance while satisfying communication quality of service (QoS) requirements. 
Specifically, the objective is to minimize the 1-bit Cram\'er–Rao bound (CRB) for estimating the azimuth angle of a point-like target under symbol-level constructive interference (CI) constraints. 
We conduct an asymptotic analysis of the 1-bit Fisher information, revealing its nonmonotonicity with the signal-to-noise ratio (SNR), and introduce amplitude constraints to exclude regions where the objective function value is clearly suboptimal and facilitate convergence to high-quality solutions. 
The resulting problem is a nonconvex optimization challenge with coupled binary and linear constraints. 
We transform the discrete problem into a continuous constrained one, characterize its global and local minima, and tackle it via the augmented Lagrangian method (ALM) and a spectral projected gradient (SPG) method combined with nonmonotone line search. 
The solution is further refined via local search and cutting-plane techniques. 
Extensive numerical experiments verify our analysis, showing that the proposed approach exhibits promising DFRC performance compared to benchmark schemes.
\end{abstract}

\begin{IEEEkeywords}
Dual-function radar-communication, multiple-input multiple-output, Cram\'er--Rao bound, 1-bit quantization, constructive interference.
\end{IEEEkeywords}

\section{Introduction}
\IEEEPARstart{T}{he} integration of radar and communication functionalities, referred to as dual-function radar-communication (DFRC) systems, represents an emerging research direction for joint hardware and spectrum sharing \cite{liu2018toward, hassanien2019dual}.
By leveraging radio-frequency (RF) signals to simultaneously perform radar sensing and communication transmission, DFRC systems can effectively support concurrent data transmission and environmental sensing capabilities, including target detection, localization, and tracking \cite{liu2022integrated}.
The implementation of DFRC systems often relies on massive multiple-input multiple-output (MIMO) techniques, which deploy a large number of antennas to improve both spectral efficiency and spatial resolution. 
However, the conventional fully digital massive MIMO architectures demand one dedicated RF chain along with a pair of digital-to-analog converters (DACs) or analog-to-digital converters (ADCs) for each antenna element. 
The power consumption of data converters (i.e., ADCs and DACs) scales exponentially with the number of quantization bits \cite{walden1999analog, allen2011cmos}. 
In addition, the high dynamic range inherent to high-resolution ADCs and DACs imposes an additional constraint on the RF chain, requiring it to maintain a large dynamic range to ensure signal integrity. 
Such high-performance RF chains are typically accompanied with high power consumption and high hardware costs.
These constraints result in prohibitive hardware complexity and power consumption for massive MIMO DFRC implementations. 

The adoption of 1-bit DACs and/or ADCs to replace their high-resolution counterparts has been recognized as a promising approach for reducing power consumption and hardware complexity in massive MIMO systems.
This technique has been extensively investigated in conventional communication systems \cite{liu2024survey, saxena2017analysis, sohrabi2018one, wang2018finite, shao2019framework, wu2024efficient2, choi2016near, mollen2017uplink, shao2024accelerated}, including the 1-bit MIMO precoding in the downlink and 1-bit channel estimation and signal detection in the uplink.
In the downlink, the main task is to design 1-bit discrete transmit signals to control the quantization distortion for optimizing the communication quality of service (QoS) metrics \cite{saxena2017analysis, sohrabi2018one, wang2018finite, shao2019framework, wu2024efficient2}; while in the uplink, it amounts to accounting for the quantization structure in developing high-performance estimation and detection algorithms \cite{choi2016near, mollen2017uplink, shao2024accelerated}.

However, the integration of 1-bit DACs and ADCs in DFRC MIMO systems poses unique challenges. 
Specifically, the 1-bit transmit signals must not only be designed to satisfy communication QoS requirements but also to simultaneously synthesize the desired sensing waveform.
Furthermore, the generated 1-bit sensing waveform will interact with the 1-bit quantization process at the receive antennas, which necessitates a joint design to address the 1-bit quantization effects at both the transmit and receive ends. This paper investigates the   precoding design problem for massive MIMO DFRC systems with both 1-bit DACs and ADCs.

\subsection{Related Works}
Most prior studies on MIMO DFRC system design assume high-resolution MIMO architectures. DFRC design methodologies can be broadly categorized into three types: communication-centric \cite{ni2023uplink, keskin2021mimo}, radar-centric \cite{nowak2016co, roberton2003, saddik2007ultra}, and joint design \cite{liu2021cramer, liu2022transmit, wen2023efficient, wu2024efficient, liu2018toward, yu2022precoding, tsinos2021joint, wu2025quantized, wang2025interference, liu2021dual, an2023fundamental, xu2024mimo}. Among them, joint design strategies optimize waveforms to support dual functions and typically exhibit better spectral efficiency or sensing performance than the other two paradigms, depending on the specific performance metrics involved.
For communication, representative metrics include signal-to-interference-plus-noise ratio (SINR) \cite{liu2021cramer, liu2022transmit, wen2023efficient}, sum rate \cite{wu2024efficient}, multiuser interference (MUI) energy \cite{liu2018toward, yu2022precoding, tsinos2021joint}, and constructive interference (CI) \cite{wu2025quantized, wang2025interference, liu2021dual}.
For radar sensing, widely adopted metrics include beampattern mean squared error (MSE) \cite{wu2025quantized, liu2021dual}, covariance matrix similarity \cite{liu2018toward, liu2022transmit, yu2022precoding}, radar signal-to-noise ratio (SNR) \cite{wen2023efficient, tsinos2021joint}, detection/false alarm probability \cite{an2023fundamental}, and Cram\'er--Rao bound (CRB) \cite{liu2021cramer, wu2024efficient, xu2024mimo}.

There is growing research interest in MIMO DFRC systems equipped with low-resolution DACs at the transmit antennas and high-resolution ADCs at the receive antennas \cite{wu2025quantized, yu2022precoding, cheng2021transmit, wang2025interference}.  
Note that many performance metrics applicable to high-resolution MIMO scenarios, such as SINR and sum rate, are no longer viable due to coarse quantization effects.
The works in \cite{yu2022precoding, cheng2021transmit} minimize MUI energy for communication, while constraining the dissimilarity between the transmit and desired sensing covariance matrices \cite{yu2022precoding} and CRB \cite{cheng2021transmit}. 
In \cite{wu2025quantized}, the authors studied the minimization of the beampattern MSE under CI-based communication QoS constraints. 
In \cite{wang2025interference}, a weighted objective function comprising the illumination power in the directions of interest and the minimum CI scaling factor is optimized. 

More recently, the works \cite{wang2023joint} and \cite{sun2026one} investigate MIMO DFRC system design with both 1-bit DACs and ADCs, seeking to approximate the quantization effect of 1-bit ADCs.
Specifically, the work \cite{wang2023joint} maximizes the probability of detection while satisfying CI-based communication QoS constraints; the nonlinear effect induced by 1-bit ADCs is approximated via a first-order Taylor expansion, which yields reasonable approximation performance at low radar SNR.
In addition, the work in \cite{sun2026one} approximates the quantization effect by additive noise and maximizes the signal-to-quantization-plus-clutter-plus-noise ratio for target sensing.

\subsection{Our Contributions}
In this paper, we investigate the design of massive MIMO DFRC systems equipped with 1-bit DACs and ADCs. 
We directly optimize the sensing CRB under 1-bit ADCs rather than relying on approximations, while ensuring the communication QoS exceeds a predefined level. 
In our formulation, we analyze the asymptotic behavior of the 1-bit Fisher information and transform the design into a more tractable counterpart. 
We also propose a dedicated optimization method to address the nonconvex design problem under coupled binary and linear constraints. 
The main contributions of this paper are summarized as follows: 

\begin{itemize}
    \item \textit{Joint 1-bit CRB and CI DFRC formulation}. 
We formulate the system design problem for massive MIMO DFRC systems with both 1-bit DACs and ADCs.
In particular, we directly tackle the CRB for estimating the target azimuth angle to address the quantization effect of 1-bit ADCs, instead of approximating the quantization effect as noise. 
We adopt the CI principle, a symbol-level precoding approach, as the communication performance metric.  
Compared with conventional block-level precoding schemes such as zero-forcing, CI enables finer-grained control over transmit signal amplitudes, which is more compatible with the binary constraints imposed by 1-bit DACs at the symbol-level. 
As a result, our design aims to optimize the 1-bit CRB under constraints on the symbol error probability (SEP) of CI, together with the binary transmit signal constraints.

	\item \textit{Asymptotic analysis of the 1-bit CRB and problem reformulation}. 
We present an asymptotic analysis of the 1-bit Fisher information, which is the inverse of the 1-bit CRB. 
The results show that the 1-bit Fisher information is not monotonically increasing with the SNR, which is different from the Fisher information without quantization. 
In light of this revelation, we impose amplitude constraints on the design problem to exclude regions where the objective function value is clearly suboptimal; this reduces the solution scope of the problem and facilitates convergence to high-quality solutions. 
Moreover, within these amplitude constraints, we show that the 1-bit Fisher information can be well approximated by the Fisher information without quantization, and the latter has a simpler quadratic form. 
In addition, we find that the overall problem can be decoupled into independent subproblem instances with smaller dimensions across the time index, which facilitates the development of computationally efficient optimization algorithms.

	\item \textit{Efficient algorithm for the resulting problem}. 
The resulting problem is to minimize a concave quadratic function subject to coupled binary and linear constraints, which poses an optimization challenge. 
We develop an efficient optimization strategy to address this challenge. 
Specifically, we transform the discrete problem into a continuous one by introducing a negative quadratic penalty term into the objective function and relaxing the binary constraints to box constraints. 
We then analyze the relationship between the penalized problem and the original formulation in terms of their global and local optima, showing that as long as the continuous penalized problem is solved to a local minimum, most elements of the solution will be binary. 
We adopt the augmented Lagrangian method (ALM) to tackle the penalized problem, where each subproblem is handled by a nonmonotone spectral projected gradient (SPG) method. Finally, the solution is refined via local search and cutting-plane techniques to better satisfy feasibility and enhance performance.
\end{itemize}

We conduct extensive numerical experiments to verify our analysis and evaluate the performance of the proposed approach. 
The results demonstrate that the introduction of amplitude constraints indeed yields solutions with superior performance, and the number of non-binary elements in the solution is consistent with our analytical bound. 
Compared with benchmark schemes, the proposed approach exhibits promising DFRC performance. 
In particular, in our test case operating at a radar SNR of 8 dB, our approach achieves sensing performance comparable to that of a system equipped with 1-bit DACs and high-resolution ADCs operating at 5 dB.

\subsection{Organization and Notations}
The rest of this paper is organized as follows. In Section \ref{section:system model}, we introduce the system model. In Section \ref{section:problem formulation}, we analyze the asymptotic behavior of the 1-bit Fisher information and discuss problem formulation. In Section \ref{section:algorithm design}, we present the proposed algorithm. In Section \ref{section:numerical results}, simulation results are provided to present the performance of the proposed design. Finally, in Section \ref{section:conclusion}, we conclude the paper.

\textit{Notations}: Scalars are denoted by Italic type (e.g., $K$), vectors by bold lowercase letters (e.g., $\bx$), and matrices by bold uppercase letters (e.g., $\bH$). The imaginary unit is denoted by $j=\sqrt{-1}$. The operators $(\cdot)^{\T}$ and $(\cdot)^{\H}$ denote the transpose and Hermitian transpose, respectively. 
The operators $\Re(\cdot)$ and $\Im(\cdot)$ return the real and imaginary parts of their arguments, respectively. 
The notation $\|\bx\|$ denotes the $\ell_2$ norm of the vector $\bx$. 
For a positive integer $n$, $[n]$ represents the set $\{1,2,\ldots,n\}$. 
The vectors $\mathbf{0}$ and $\mathbf{1}$ denote the all-zero vector and the all-one vector of appropriate dimensions, respectively, and $\mathbf{I}$ denotes the identity matrix of appropriate dimension. 
For vectors $\bx$ and $\by$, $\max\{\bx,\by\}$ ($\min\{\bx,\by\}$) denotes the component-wise maximum (minimum) of the two vectors. For a vector $\bx$ and a closed convex set $\mathcal{X}$, $\mathcal{P}_{\mathcal{X}}(\bx)$ denotes the projection of $\bx$ onto $\mathcal{X}$. Finally, $\mathcal{CN}(\mathbf{0},\bf \Sigma)$ denotes the zero-mean circularly symmetric complex Gaussian distribution with covariance matrix $\bf \Sigma$.

\section{System Model}
\label{section:system model}

\begin{figure}[t]
	\centering
	\includegraphics[width=.45\textwidth]{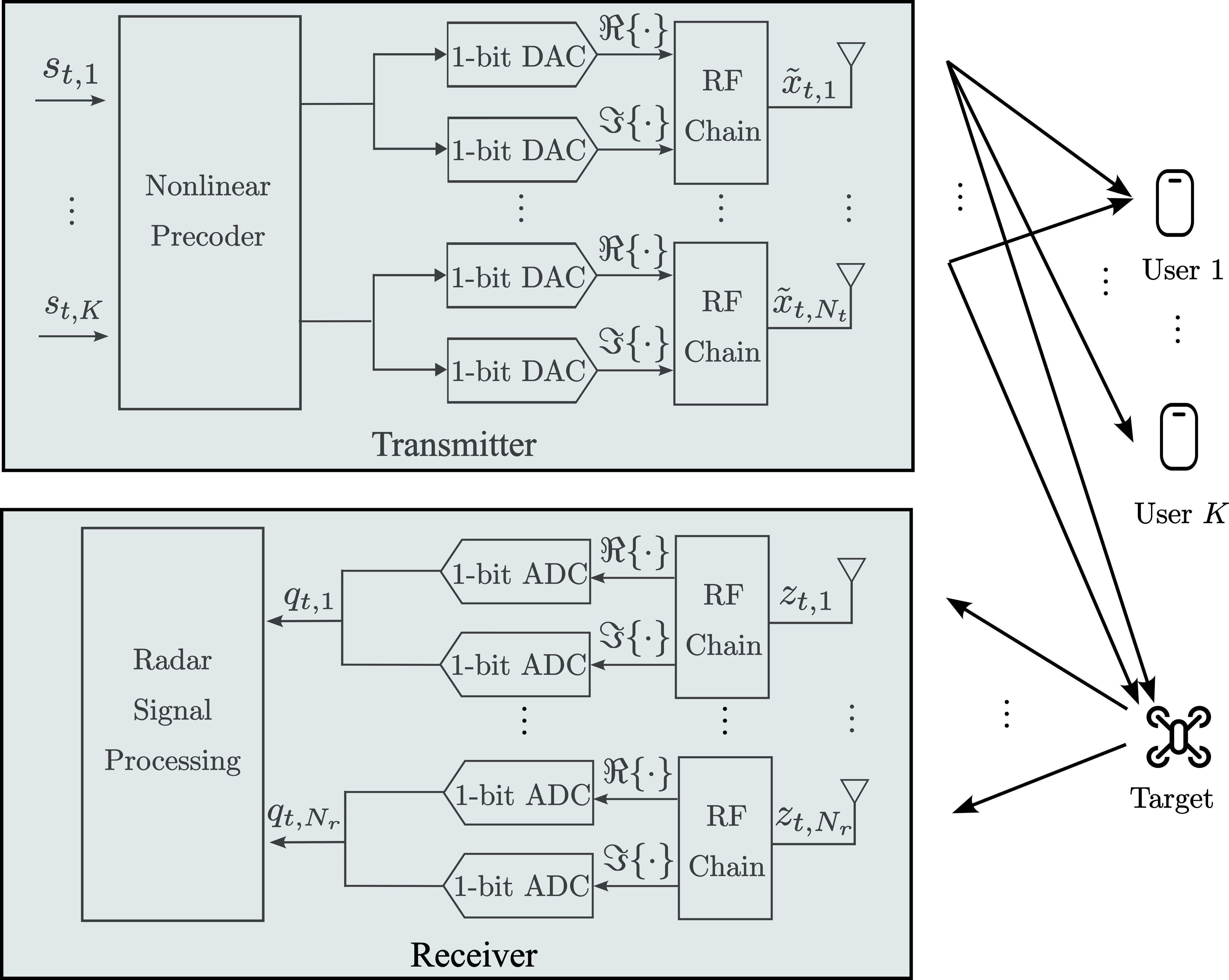}
	\caption{A massive MIMO DFRC system equipped with 1-bit DACs and ADCs.}
	\label{fig:MIMO DFRC system}
\end{figure}

We consider a  massive MIMO DFRC system, where a base station (BS) equipped with $N_t$ transmit antennas and $N_r$ receive antennas simultaneously serves $K$ single-antenna communication users while sensing a point-like radar target over the same time-frequency resources, as illustrated in Fig.~\ref{fig:MIMO DFRC system}. Both the transmit and receive antennas are arranged as uniform linear arrays (ULAs) with an inter-antenna spacing of half the carrier wavelength. Each transmit antenna is equipped with a pair of 1-bit DACs, and each receive antenna is equipped with a pair of 1-bit ADCs.

\subsection{Communication Model}
Let $\bX=[\tilde{\bx}_1,\tilde{\bx}_2,\ldots,\tilde{\bx}_T]\in \C^{N_t\times T}$ denote the dual-functional transmit signal matrix over a block of length $T$, where $\tilde{\bx}_t$ is the transmit signal vector at time slot $t$. Let $P$ denote the total transmit power budget at the BS. To comply with the 1-bit DAC hardware constraints, the transmit signal is restricted to the following discrete set:
\begin{equation*}
	\tilde{\bx}_{t} \in \mathcal{\tilde{X}} := \left\{\pm \eta \pm j \eta  \right\}^{N_t}, \forall~t \in [T],
\end{equation*}
where $\eta=\sqrt{P/(2N_t)}$.

The received signal at user $k$ at time slot $t$ is given by
\begin{equation*}
	y_{t,k}= \bh_k^{\T} \tilde{\bx}_t + n_{t,k},
\end{equation*}
where $\bh_k\in\C^{N_t}$ denotes the downlink channel vector from the BS to user $k$, and $n_{t,k}\sim\mathcal{CN}(0,\sigma_c^2)$ represents the additive white Gaussian noise (AWGN).

The design objective for communication is to minimize the  SEP  at the multiuser sides. 
Let $s_{t,k}$ denote the intended data symbol for user $k$ at time slot $t$. 
We assume that the user symbols $s_{t,k}$'s are drawn from an $M$-PSK constellation, and the communication users adopt nearest-neighbor decoding.  
To harness the structure of the $M$-PSK constellation, we adopt the CI principle \cite{masouros2015exploiting}, whereby the received signal $y_{t,k}$ is deliberately pushed deeper into the decision region of $s_{t,k}$.  

\begin{figure}[t]
	\centering
	\includegraphics[width=.36\textwidth]{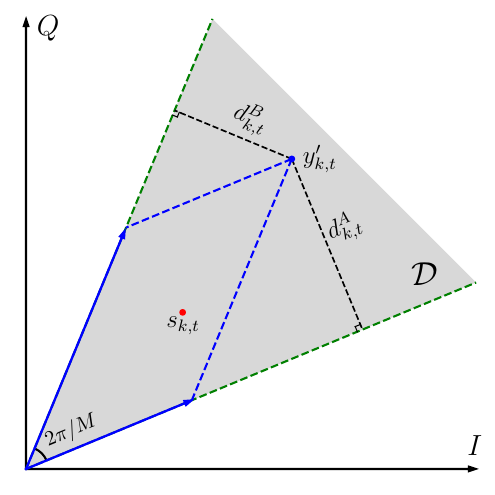}
	\caption{An illustration of constructive interference, where the shaded area $\cal D$ denotes the decision region of intended symbol.}
	\label{fig:CI}
\end{figure}

We illustrate the concept of CI in Fig.~\ref{fig:CI} using an example with $M=8$. Let $y'_{t,k}=\bh_k^{\T} \tilde{\bx}_t$ denote the noise-free received signal at user $k$. 
It has been shown in \cite{wu2024efficient2, wu2023diversity} that the distances $d_{t,k}^A, d_{t,k}^B$ between $y'_{t,k}$ and the decision boundaries of the symbol $s_{t,k}$ are related to the SEP, defined as $\mathrm{SEP}_{t,k}=\mathrm{P}(y_{t,k} \notin \mathcal{D}(s_{t,k})|s_{t,k})$, which is shown as follows:
	\begin{equation}
		\label{eq:SEP upper bound}
		\frac{1}{2}\mathrm{erfc}\left( \frac{d_{t,k}}{\sigma_c} \right) \leq \mathrm{SEP}_{t,k} \leq \mathrm{erfc}\left(\frac{d_{t,k}}{\sigma_c} \right),
	\end{equation}
where $d_{t,k} = \min\{d_{t,k}^A, d_{t,k}^B\}$ and $ \mathrm{erfc}(z)=1-\frac{2}{\sqrt{\pi}} \int_{0}^{z} e^{-t^2}\mathrm{d}t$ 
is the complementary error function. 

According to \eqref{eq:SEP upper bound}, maintaining the $\mathrm{SEP}_{t,k}$ below $\epsilon$ can be ensured by enforcing $d_{t,k} \geq \sigma_c\text{erfc}^{-1}(\epsilon)$. Thus, to guarantee that all users' SEPs do not exceed $\epsilon$, we can impose the following constraint:
\begin{equation}
	\label{eq:CI constraint}
	\bC_{t}\bx_t \geq \boldsymbol{\gamma},~\forall~t\in [T],
\end{equation}
where 
\begin{equation*}
	\bC_{t,k} = \begin{bmatrix}
			\Im(s_{t,k}^B) & -\Re(s_{t,k}^B) \\
			-\Im(s_{t,k}^A) & \Re(s_{t,k}^A)
		\end{bmatrix}
		\begin{bmatrix}
			\Re(\bh_k^{\T}) & -\Im(\bh_k^{\T}) \\
			\Im(\bh_k^{\T}) & \Re(\bh_k^{\T})
		\end{bmatrix},
\end{equation*}
$\boldsymbol{\gamma}=\sigma_c \text{erfc}^{-1}(\epsilon) \cdot\mathbf{1}$ and $\bx_t=[\Re(\tilde{\bx})^{\T}, \Im(\tilde{\bx})^{\T}]^{\T}$ is the real representation of $\tilde{\bx}_t$.

\subsection{Radar Model}
We assume line-of-sight (LOS) propagation between the radar target and the receive antennas. With the transmission signal $\tilde{\bx}_{t}$, the received signal under 1-bit ADCs can be expressed as
\begin{equation*}
	\bq_t=\mathcal{Q}(\bz_t),~\bz_t=\beta \ba_r\ba_t^{\H}\tilde{\bx}_t+\bw_t,
\end{equation*}
where $\mathcal{Q}(\cdot)=\sign(\Re(\cdot))+j\sign(\Im(\cdot))$ is the 1-bit quantization function, $\beta\in \C$ is the radar reflection coefficient, $\ba_r~\text{and}~\ba_t$ are the steering vectors of receive and transmit antennas, respectively, with
$$
	\ba_r=\left[ e^{-j\frac{N_r-1}{2}\pi\sin(\theta)}, e^{-j\frac{N_r-3}{2}\pi\sin(\theta)},\ldots,e^{j\frac{N_r-1}{2}\pi\sin(\theta)} \right]^\mathrm{T},
$$
$$
	\ba_t=\left[ e^{-j\frac{N_t-1}{2}\pi\sin(\theta)}, e^{-j\frac{N_t-3}{2}\pi\sin(\theta)},\ldots,e^{j\frac{N_t-1}{2}\pi\sin(\theta)} \right]^\mathrm{T},
$$
$\theta$ is the azimuth angle of target, and $\bw_t\sim \mathcal{CN}(\mathbf{0},\sigma_r^2\mathbf{I})$ is the AWGN at the receive antennas.

The radar reflection coefficient $\beta$ is usually assumed to remain constant within the transmit block (i.e., following Swerling II model \cite{skolnik2002introduction}), and can be estimated using preprocessing techniques \cite{deng2022receive, xi2020gridless}. 
Then, the sensing task amounts to estimate the angle $\theta$, and we resort to CRB minimization. 
To this end, we derive  the log-likelihood function for estimating $\theta$ under 1-bit quantized observations $\bq_1,\bq_2,\ldots,\bq_T$ as follows:
\begin{equation*}
	\begin{aligned}
		l_{\text{1-bit}}(\theta|\bq_1,\bq_2,\ldots,\bq_T)=&\sum_{t=1}^T \sum_{n=1}^{N_r} \left[\log\Phi\left(\frac{\sqrt{2}}{\sigma_r}\Re(q_{t,n})\Re(r_{t,n}) \right)\right. \\
		&~+ \left.\log\Phi\left(\frac{\sqrt{2}}{\sigma_r}\Im(q_{t,n})\Im(r_{t,n}) \right) \right],
	\end{aligned}
\end{equation*}
where $\mathbf{r}_t=\beta \ba_r\ba_t^\mathrm{H}\tilde{\bx}_t$, and $\Phi(z)=\frac{1}{\sqrt{2\pi}}\int_{-\infty}^{z} e^{-\frac{t^2}{2}}\mathrm{d}t $ is the cumulative density function of the standard normal distribution. 
With the log-likelihood function $l_{\text{1-bit}}$, the Fisher information for estimating $\theta$ can be derived in a manner similar to that in \cite{stoica2021cramer}, yielding
\begin{equation}
	\label{eq:1-bit FI}
	\begin{aligned}
	F_{\textnormal{1-bit}}(\bX)=&\mathbb{E}\left[\left( \frac{\partial l_{\text{1-bit}}(\theta|\bq_1,\bq_2,\ldots,\bq_T)}{\partial\theta} \right)^2\right] \\
	=&\frac{2}{\sigma_r^2} \sum_{t=1}^T \sum_{n=1}^{N_r} \Bigg[\rho\left(\frac{\sqrt{2}}{\sigma_r}\Re(r_{t,n})\right)\left(\frac{\partial \Re(r_{t,n})}{\partial \theta}\right)^2 \\
	&+ \rho\left(\frac{\sqrt{2}}{\sigma_r}\Im(r_{t,n})\right) \left(\frac{\partial \Im(r_{t,n})}{\partial \theta}\right)^2 \Bigg],
	\end{aligned}	
\end{equation}
where 
\begin{equation}
	\label{eq:rho function}
	\rho(u)=\frac{\phi^2(u)}{\Phi(u)[1-\Phi(u)]},
\end{equation}
with $\phi(z) = \frac{1}{\sqrt{2\pi}}e^{-\frac{z^2}{2}} $ being the probability density function of the standard normal distribution.

For comparison, the Fisher information for estimating $\theta$ under unquantized observations $\bz_1,\bz_2,\ldots,\bz_T$ is given by
\begin{equation*}
	\begin{aligned}
		& F_{\textnormal{unquant.}}(\bX) \\
		& = \frac{2}{\sigma_r^2}\sum_{t=1}^T\sum_{n=1}^{N_r}\left[\left( \frac{\partial \Re(r_{t,n}) }{\partial \theta} \right)^2 + \left( \frac{\partial \Im(r_{t,n}) }{\partial \theta} \right)^2\right] \\
		& = \frac{2}{\sigma_r^2}\sum_{t=1}^T \|\bA\bx_t \|^2,
	\end{aligned}
	\label{eq:unquantize FI}
\end{equation*}
where  
\begin{equation}
	\label{eq:matrix A}
	\bA=\begin{bmatrix}
	\Re(\tilde{\bA}) & -\Im(\tilde{\bA}) \\
	\Im(\tilde{\bA}) & \Re(\tilde{\bA})
\end{bmatrix}
\end{equation}
with $\tilde{\bA}=\beta(\dot{\ba}_r\ba_t^{\H}+\ba_r\dot{\ba}_t^{\H}), \dot{\ba}_r= \mathrm{d}\ba_r/\mathrm{d}\theta$ and $\dot{\ba}_t=\mathrm{d}\ba_t/\mathrm{d}\theta$. 
Hereafter, we refer to the Fisher information $F_{\textnormal{unquant.}}$ as the \emph{unquantized Fisher information}, and the Fisher information $F_{\textnormal{1-bit}}$ as the \emph{1-bit Fisher information}.

Two remarks are in order. First, by defining $\boldsymbol{\kappa}(\bx_t) \in \R^{2N_r}$ with
$$
\left[\boldsymbol{\kappa}(\bx_t)\right]_{n}= \left\{
\begin{aligned}
	&\rho\left(\frac{\sqrt{2}}{\sigma_r}\Re\left(r_{t,(n+1)/2}\right)\right),~\text{if}~n~\text{is odd}; \\
	&\rho\left(\frac{\sqrt{2}}{\sigma_r}\Im\left(r_{t,n/2}\right)\right),~\text{if}~n~\text{is even},
\end{aligned}
\right.
$$
for $n = 1,2,\dots, 2N_r$, as the scaling coefficient vector that encompasses the scaling coefficients in \eqref{eq:1-bit FI}, and defining $\mathbf{f}(\bx_t) \in \R^{2N_r}$ with
$$
	\left[\mathbf{f}(\bx_t)\right]_n=\left\{
	\begin{aligned} &\left( \frac{\partial \Re\left(r_{t,(n+1)/2}\right) }{\partial \theta} \right)^2,~\text{if}~n~\text{is odd};  \\
		&\left( \frac{\partial \Re\left(r_{t,n/2}\right) }{\partial \theta} \right)^2,~\text{if}~n~\text{is even},
	\end{aligned}\right. 
$$
for $n = 1,2,\dots, 2N_r$, as the Fisher information vector that encompasses the unquantized Fisher information, we can rewrite $F_{\textnormal{1-bit}}(\bX)$ and $F_{\textnormal{unquant.}}(\bX)$ as follows:
\begin{equation}\label{eq:Fisher_com}
	F_{\textnormal{1-bit}}(\bX)=\frac{2}{\sigma_r^2}\sum_{t=1}^T \langle \boldsymbol{\kappa}(\bx_t), \mathbf{f}(\bx_t) \rangle,
\end{equation}
\begin{equation}
	\label{eq:inner product representation}
	F_{\textnormal{unquant.}}(\bX)=\frac{2}{\sigma_r^2}\sum_{t=1}^T \langle \mathbf{1}, \mathbf{f}(\bx_t) \rangle,
\end{equation}
which means the 1-bit Fisher information can be viewed as a \emph{scaled} version of the unquantized Fisher information. 

Second, according to the expression of $\rho(\cdot)$, it is a positive-valued even function that attains its maximum value of $2/\pi$ at $0$. By comparing \eqref{eq:Fisher_com} and \eqref{eq:inner product representation}, we obtain
\begin{equation}
	F_{\text{1-bit}}(\bX) \leq \frac{2}{\pi} F_{\text{unquant.}}(\bX).
	\label{eq:1-bit FI upper bound}
\end{equation}
In other words, 1-bit quantization incurs at least a $-10\log_{10}(2/\pi) \approx 1.96$~dB loss in the Fisher information, a phenomenon that has been reported in literature \cite{host2000effects, xi2020gridless, stoica2021cramer}.

\section{Problem Formulation and Analysis}
In this section, we propose a CRB-optimized formulation for massive MIMO DFRC systems with 1-bit DACs and ADCs, and analyze the asymptotic behavior of the 1-bit Fisher information.

\label{section:problem formulation}
\subsection{Problem Formulation}
Our objective is to minimize the 1-bit CRB, while satisfying the communication QoS constraint and the 1-bit transmit signal constraint. Since the 1-bit CRB is the inverse of the 1-bit Fisher information, minimizing the 1-bit CRB is equivalent to maximizing the 1-bit Fisher information. 
As a result, the MIMO DFRC system design problem under both 1-bit DACs and ADCs is formulated as
\begin{subequations}
	\label{eq:formulation}
	\begin{align}
		\max_{\bX} & \quad F_{\text{1-bit}}(\bX) \\
		\st & \quad \bC_t \bx_t \geq \boldsymbol{\gamma}, \quad \forall~t \in [T], \label{eq:formulation comm constraint} \\
		& \quad \bx_t \in \{-\eta, \eta\}^{2N_t}, \quad \forall~t \in [T]. \label{eq:formulation binary constraint}
	\end{align}
\end{subequations}
Problem \eqref{eq:formulation} is challenging for the following reasons. First, the objective function $F_{\text{1-bit}}(\bX)$ is nonlinear and nonconvex. Second, the problem is a large-scale integer programming problem, for which the complexity of exhaustive search is $\mathcal{O}(2^{2N_t T})$; in particular, the number of antennas $N_t$ can reach the order of hundreds in massive MIMO systems. Third, the coupling effect of constraints \eqref{eq:formulation comm constraint} and \eqref{eq:formulation binary constraint} further complicates the algorithm design.

\subsection{Asymptotic Analysis and Problem Reformulation}

In this subsection, we analyze the behavior of the function $F_{\text{1-bit}}$ and delineate the possible solution scope of the design problem \eqref{eq:formulation}, aiming for a more tractable and regularized transformation of the design problem.

As shown in \eqref{eq:Fisher_com}, the 1-bit Fisher information is an inner product of two parts, the scaling coefficient part $\boldsymbol{\kappa}(\bx_t)$ and Fisher information part $\mathbf{f}(\bx_t)$. 

We first investigate the properties of $\rho(\cdot)$  in the scaling coefficient vector $\boldsymbol{\kappa}(\bx_t)$. As shown in Fig.~\ref{fig:rho_plot}, $\rho(\cdot)$ is an even, strictly positive function that decreases monotonically on $[0,\infty)$, and has a vanishing tail, i.e., $\rho(u)\rightarrow 0$ and $\rho'(u) \rightarrow 0$ as $u\rightarrow \infty$. 
In fact, it can be shown that 
$$
	\lim_{|u|\rightarrow \infty} \frac{\rho(u)}{e^{-u^2/4}}=0.
$$ 
In other words, $\rho(u)$ decays faster than $e^{-u^2/4}$ as $|u|\rightarrow \infty$. We can further show the following proposition. 
\begin{proposition}
	\label{proposition:asymptotic 1-bit FI}
	For any transmit signal matrix $\bX$ satisfying 
	$\Re(r_{t,n}) \neq 0$ and $\Im(r_{t,n}) \neq 0$ for all 
	$t \in [T]$ and $n \in [N_r]$, it holds that
	\[
	\lim_{ \textnormal{SNR}_{\textnormal{radar}} \to +\infty} F_{\mathrm{1\text{-}bit}}(\bX) = 0,
	\]
	where $\textnormal{SNR}_{\textnormal{radar}} := |\beta|^2 P / \sigma_r^2$ denotes the radar SNR.
\end{proposition} 
\begin{proof}
	See Appendix \ref{section:proof of proposition asymptotic}.
\end{proof}
Proposition~\ref{proposition:asymptotic 1-bit FI} states that under mild conditions, as the radar SNR increases to infinity, the 1-bit Fisher information for fixed transmit signal matrix will tend to zero. 
\begin{figure}[t]
	\centering 
	\includegraphics[width=2.8in]{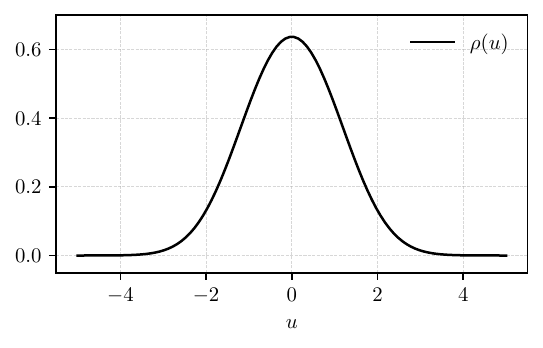}
	\caption{Plot of $\rho(u)$.}
	\label{fig:rho_plot}
\end{figure}
\begin{figure}[t]
	\centering
	\includegraphics[width=3.2in]{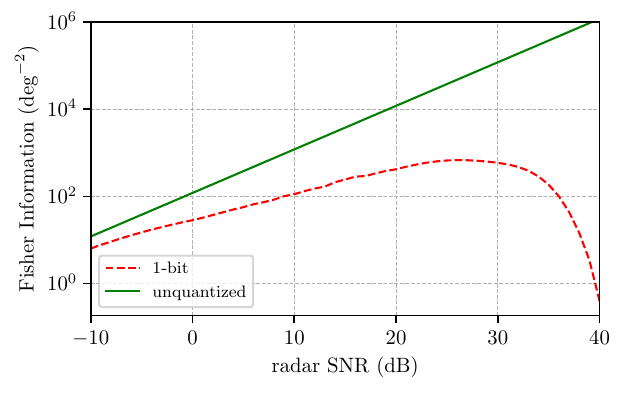}
	\caption{Asymptotic behavior of the 1-bit Fisher information.}
	\label{fig:asymptotic FI}
\end{figure}
Proposition~\ref{proposition:asymptotic 1-bit FI} is illustrated in Fig.~\ref{fig:asymptotic FI}, where we compare the 1-bit Fisher information with the unquantized Fisher information for a MIMO DFRC system equipped with $N_t=10$ transmit antennas and $N_r=10$ receive antennas. The transmit signal matrix is obtained by solving the following optimization problem via exhaustive search:
\begin{subequations}
	\label{eq:formulation unquantized}
	\begin{align}
		\max_{\bX} & \quad F_{\text{unquant.}}(\bX) \\
		\st & \quad \bC_t \bx_t \geq \boldsymbol{\gamma}, \quad \forall~t \in [T], \\
		& \quad \bx_t \in \{-\eta, \eta\}^{2N_t}, \quad \forall~t \in [T]. 
	\end{align}
\end{subequations}

According to the analysis above, when the radar SNR is high, the scaling coefficients in $\boldsymbol{\kappa}(\bx_t)$ approach zero, resulting in a low Fisher information $F_\text{1-bit}$. This is contrary to our goal of achieving high Fisher information in problem \eqref{eq:formulation}.  
Furthermore, from the perspective of nonconvex optimization, eliminating regions where the function value is clearly suboptimal can  reduce the solution scope of the problem and facilitate convergence to high-quality solutions.
As a result, we impose amplitude constraints to prevent solutions of low Fisher information.
Specifically, we restrict the amplitudes of the arguments in $\rho$ via
\begin{equation}
	\label{eq:amplitude constraint}
	A_n(\bx_t) \leq \zeta, \quad \forall~t \in [T],~ \forall~n \in [N_r],
\end{equation}
where
$$
A_n(\bx_t)=\max\left\{ \left|\frac{\sqrt{2}}{\sigma_r}\Re(r_{t,n})\right|, \left|\frac{\sqrt{2}}{\sigma_r}\Im(r_{t,n})\right| \right\}
$$
and $\zeta$ is a parameter that controls the restriction range. The choice of $\zeta$ involves a tradeoff between the scaling coefficient term $\boldsymbol{\kappa}(\bx_t)$ and the Fisher information term $\mathbf{f}(\bx_t)$. A moderate $\zeta$ imposes tighter amplitude constraints, thereby achieving a larger $\boldsymbol{\kappa}(\bx_t)$; however, the Fisher information term $\mathbf{f}(\bx_t)$ may deteriorate if the $\zeta$ is too small to  restrict the feasible region. 

Note that the amplitude constraint  \eqref{eq:amplitude constraint} is composed of $4N_rT$ linear inequality constraints. We show that these $4 N_r T$ linear inequality constraints can be approximated by $4T$ linear inequality constraints, which effectively removes the dependence on the number of receive antennas $N_r$. Consider the following proposition.
\begin{proposition}
\label{proposition:reduction of constraints}
For any $t \in [T]$, define
\begin{equation*}
\mathcal{A}(\zeta) = \Bigg\{ \bx_t \in \{-\eta, \eta\}^{2N_t} \,\Bigg|\,
A_n(\bx_t) \le \zeta, \forall~n \in [N_r] \Bigg\},
\end{equation*}
and
\begin{equation*}
\mathcal{B}(\zeta) = \Bigg\{ \bx_t \in \{-\eta, \eta\}^{2N_t} \,\Bigg|\,
A_1(\bx_t) \le \zeta
\Bigg\}.
\end{equation*}
Then, it holds that
\begin{equation*}
\mathcal{A}(\zeta) \subseteq \mathcal{B}(\zeta) \subseteq \mathcal{A}(\sqrt{2}\,\zeta).
\end{equation*}
\end{proposition}
\begin{proof}
	See Appendix \ref{section:proof of proposition reduction}.
\end{proof}

Proposition \ref{proposition:reduction of constraints} shows that for each $t \in [T]$, we can approximate the $N_r$ amplitude constraint in \eqref{eq:amplitude constraint} by a single constraint
\begin{equation}
	\label{eq:amplitude constraint reduced}
	A_1(\bx_t) \leq \zeta.
\end{equation}
Thus, we incorporate constraint \eqref{eq:amplitude constraint reduced} into problem \eqref{eq:formulation} to refine the  solution space.

On the other hand, constraint \eqref{eq:amplitude constraint reduced} brings a new opportunity for a simple approximation of the objective function in problem \eqref{eq:formulation}. Under constraint \eqref{eq:amplitude constraint reduced}, the 1-bit Fisher information can be both upper and lower bounded by the unquantized Fisher information, which is shown in Proposition~\ref{proposition:well approximation}. 
\begin{proposition}
\label{proposition:well approximation}
	For any transmit signal matrix $\bX$ satisfying \eqref{eq:formulation comm constraint}, \eqref{eq:formulation binary constraint}, and \eqref{eq:amplitude constraint reduced}, it holds that
		\begin{equation}
		\label{eq:obj approximation}
			\rho(\sqrt{2}\zeta) F_{\textnormal{unquant.}}(\bX) \leq F_{\textnormal{1-bit}}(\bX) \leq \frac{2}{\pi} F_{\textnormal{unquant.}}(\bX).
		\end{equation}
\end{proposition}
\begin{proof}
	See Appendix \ref{section:proof of proposition well approximation}.
\end{proof}

Proposition~\ref{proposition:well approximation} provides a design hint: given that the unquantized Fisher information has a much simpler form (without the complex scaling coefficients $\boldsymbol{\kappa}(\bx_t)$), we propose to tackle the unquantized Fisher information as a surrogate for  the 1-bit Fisher information $F_{\text{1-bit}}$. 

To summarize, with amplitude constraint and objective change, the design problem \eqref{eq:formulation} is then transformed into 
\begin{subequations}
	\label{eq:formulationun_ac}
	\begin{align}
		\max_{\bX} & \quad F_{\text{unquant.}}(\bX) \\
		\st & \quad \eqref{eq:formulation comm constraint}, \eqref{eq:formulation binary constraint}, \eqref{eq:amplitude constraint reduced}.
	\end{align}
\end{subequations}
Since both the unquantized Fisher information $ F_{\text{unquant.}}(\bX)$ and constraints in problem \eqref{eq:formulationun_ac} are separable across $t$, we can decompose the problem \eqref{eq:formulationun_ac} into $T$ independent subproblems, each corresponding to a single time slot. We omit the time index $t$ and express the independent design problem for each time slot as
\begin{subequations}
	\label{eq:original problem}
	\begin{align}
		\min_{\bx} & \quad f(\bx) \\
		\st & \quad \bC \bx \geq \boldsymbol{\gamma}, \label{eq:original problem comm constraint} \\
		& \quad \bD \bx \leq \boldsymbol{\zeta}, \label{eq:original problem amplitude constraint}\\
		& \quad \bx \in \{-\eta, \eta\}^{2N_t}, \label{eq:original problem binary constraint}
	\end{align}
\end{subequations}
where $f(\bx)= -\|\bA\bx \|^2$, $\boldsymbol{\zeta} = \zeta \mathbf{1}$, and
\begin{equation*}
	\bD=\frac{\sqrt{2}}{\sigma_r} \begin{bmatrix}
	\Re(\beta a_{r,1} \ba_t)^{\T} & \Im(\beta a_{r,1}\ba_t)^{\T} \\
	-\Im(\beta a_{r,1}\ba_t)^{\T} & \Re(\beta a_{r,1}\ba_t)^{\T} \\
	-\Re(\beta a_{r,1}\ba_t)^{\T} & -\Im(\beta a_{r,1}\ba_t)^{\T} \\
	\Im(\beta a_{r,1}\ba_t)^{\T} & -\Re(\beta a_{r,1}\ba_t)^{\T}
	\end{bmatrix} \in \R^{4\times 2N_t}.
\end{equation*}

\section{Algorithm Design}
\label{section:algorithm design}

In this section, we develop an efficient algorithm for solving problem \eqref{eq:original problem}. First, we employ a penalty technique to tackle the binary constraint \eqref{eq:original problem binary constraint} and analyze the relationship between the penalized problem and the original problem in terms of global and local minimizers in Section \ref{subsection:penalty method}. Then, we customize the ALM for solving the penalized formulation, with each ALM subproblem solved by a nonmonotone SPG method in Section \ref{subsection:ALM}. At last, in Section \ref{subsection:refinement}, we discuss techniques for refining the final solutions.
\subsection{Penalty Method}
\label{subsection:penalty method}
The binary constraint \eqref{eq:original problem binary constraint} and its coupling with the linear constraints \eqref{eq:original problem comm constraint}--\eqref{eq:original problem amplitude constraint} poses significant challenges for optimization. To address this issue, we adopt a penalty-based approach, see \cite{shao2019framework, liu2024extreme} and the overview paper \cite{liu2024survey}. By introducing a negative squared $\ell_2$-norm penalty term $-\lambda \|\bx\|^2$, where $\lambda > 0$ is the penalty parameter, into the objective function and relaxing the binary constraint to a box constraint, the original discrete optimization problem is transformed into a continuous one. The penalty term drives the solution toward the boundary of $[-\eta, \eta]^{2N_t}$, thereby leading to binary solution in $\{-\eta, \eta\}^{2N_t}$. Specifically, the penalized problem is 
\begin{subequations}
	\label{eq:penalty problem}
	\begin{align}
		\min_{\bx} & \quad f(\bx)-\lambda\|\bx\|^2 \\
		\st & \quad \bC \bx \geq \boldsymbol{\gamma}, \label{eq:penalty problem comm constraint} \\
		& \quad \bD \bx \leq \boldsymbol{\zeta}, \label{eq:penalty problem amplitude constraint} \\
		& \quad \bx \in [-\eta, \eta]^{2N_t}. \label{eq:penalty problem box constraint}
	\end{align}
\end{subequations}
Problems \eqref{eq:original problem} and \eqref{eq:penalty problem} share the same globally optimal solutions, which is stated in the following Proposition. 
\begin{proposition}[\hspace{-0.01mm}{\cite[Proposition~1]{wu2025quantized}}]
	\label{proposition:global}
	There exists $\lambda_0>0$ such that for all $\lambda > \lambda_0$, problem \eqref{eq:original problem} and \eqref{eq:penalty problem} share the same global minimizers.
\end{proposition}
Next, we analyze the property of the local minimizers. When there are only binary constraints, it has been shown that any local minimizer of the corresponding penalized problem must lie in the set $\{-\eta, \eta\}^{2N_t}$~\cite{shao2019framework, liu2024extreme}. In contrast, our problem additionally includes the linear inequality constraints \eqref{eq:penalty problem comm constraint} and \eqref{eq:penalty problem amplitude constraint}, which interact with the box constraint in \eqref{eq:penalty problem box constraint} and may introduce local minimizers outside $\{-\eta, \eta\}^{2N_t}$. To further investigate this issue, we show that any local minimizer of problem~\eqref{eq:penalty problem} must still have a large proportion of its elements in $\{-\eta, \eta\}$, as stated in Theorem \ref{theorem: Ix lower bound}.

\begin{theorem}
	\label{theorem: Ix lower bound}
	Assume that $\lambda>L_f/2$, where $L_f$ is the Lipschitz constant of $\nabla f(\cdot)$ on $[-\eta, \eta]^{2N_t}$. 
Then, any local minimizer $\hat{\bx}$ to problem \eqref{eq:penalty problem} must satisfy
	\begin{equation}
		\label{eq:Ix lower bound}
		I(\hat{\bx}) \leq 2K+2,
	\end{equation}
where $I(\bx)$ is the cardinality of non-binary elements in $\bx$, defined as 
	\[
		I(\bx)=2N_t - \sum_{n=1}^{2N_t} \mathbf{1}_{\{-\eta, \eta\}}(x_n).
	\]
\end{theorem}
\begin{proof}
	See Appendix \ref{section:proof of theorem Ix lower bound}.
\end{proof}

Theorem~\ref{theorem: Ix lower bound} establishes that any local minimizer contains at most $2K+2$ elements that do not belong to $\{-\eta, \eta\}$. In massive MIMO systems, the number of users $K$ is typically much smaller than the number of transmit antennas $N_t$. This result provides an important algorithmic insight for obtaining high-quality solutions, which will be exploited in  Section~\ref{subsection:refinement}.

\subsection{ALM for Solving \eqref{eq:penalty problem}}
\label{subsection:ALM}
In this subsection, we employ the ALM \cite{andreani2007on} to solve \eqref{eq:penalty problem}. 
We define the partial augmented Lagrangian function with respect to the constraints \eqref{eq:penalty problem comm constraint} and \eqref{eq:penalty problem amplitude constraint} as follows:
\begin{equation*}
	\begin{aligned}
		& L_{\rho_u, \rho_v}(\bx,\bu,\bv) \\
		&=f(\bx)-\lambda\|\bx\|^2+\frac{\rho_u}{2}\left\|\max\left\{-\bC\bx+\boldsymbol{\gamma}+\frac{\bu}{\rho_u}, \mathbf{0}\right\}\right\|^2 \\
		&~~~+\frac{\rho_v}{2}\left\|\max\left\{\bD\bx-\boldsymbol{\zeta}+\frac{\bv}{\rho_v}, \mathbf{0}\right\}\right\|^2,
	\end{aligned}
\end{equation*}
where $\bu$ and $\bv$ are the Lagrange multipliers associated with constraints~\eqref{eq:penalty problem comm constraint} and~\eqref{eq:penalty problem amplitude constraint}, respectively, and $\rho_u$ and $\rho_v$ are the corresponding ALM penalty parameters, respectively.

The ALM consists of a sequence of minimizations of $L(\bx,\bu,\bv,\rho_u,\rho_v)$ subject to constraint~\eqref{eq:penalty problem box constraint}, followed by updates of multipliers $\bu$, $\bv$ and penalty parameters $\rho_u$, $\rho_v$. 
To be more specific, at $\ell$-th iteration, the ALM algorithm performs three steps until convergence:
\begin{figure*}[t]
\begin{equation}
	\label{eq:rho_u rho_v update}
	\begin{aligned}
		\rho_u^{(\ell+1)}&=\left\{\begin{aligned}
			\iota\rho_u^{(\ell+1)},&~\text{if}~\sigma_1(\bx^{(\ell+1)}, \bu^{(\ell+1)}, \rho_u^{(\ell+1)}) > \tau\sigma_1(\bx^{(\ell)}, \bu^{(\ell)}, \rho_u^{(\ell)}); \\
			\rho_u^{(\ell+1)},&~\text{otherwise}.
		\end{aligned}\right. \\
		\rho_v^{(\ell+1)}&=\left\{\begin{aligned}
			\iota\rho_v^{(\ell+1)},&~\text{if}~\sigma_2(\bx^{(\ell+1)}, \bv^{(\ell+1)}, \rho_v^{(\ell+1)}) > \tau\sigma_2(\bx^{(\ell)}, \bv^{(\ell)}, \rho_v^{(\ell)}); \\
			\rho_v^{(\ell+1)},&~\text{otherwise}.
		\end{aligned}\right.
	\end{aligned}
\end{equation}
 \hrule
 \vspace{-0.3cm}
\end{figure*}

\textit{Step 1 (Solve ALM subproblem)} It requires to solve the ALM subproblem
\begin{equation}
	\label{eq:ALM subproblem}
	\min_{\bx \in [-\eta, \eta]^{2N_t}} f^{(\ell)}(\bx):= L_{\rho_u^{(\ell)}, \rho_v^{(\ell)}}(\bx,\bu^{(\ell)},\bv^{(\ell)}),
\end{equation}
which is a nonconvex optimization problem with a smooth objective function and a convex box constraint. 

One can apply off-the-shelf optimization methods, such as the projected gradient method \cite{calamai1987projected}, active set method \cite{hager2006new}, or trust region method \cite{friedlander1994new} to address this problem.
Here, we adopt a nonmonotone SPG method to solve \eqref{eq:ALM subproblem}. 
This method employs an advanced step-size strategy inspired by Newton's method, known as the Barzilai--Borwein (BB) steplength \cite{barzilai1988two}, and typically exhibits superior solution quality and convergence speed compared to the conventional projected gradient method.

The nonmonotone SPG method works as follows. 
Let $\left\{x^{[m]}\right\}$ denote the sequence generated by the SPG method, with $[m]$ denoting the SPG iteration. 
At the $m$-th iteration, $\bx^{[m]}$ is updated via
\begin{equation}
	\label{eq:SPG update}
	\bx^{[m+1]}=\mathcal{P}_{[-\eta, \eta]^{2N_t}}\left(\bx^{[m]}-\delta \alpha^{[m]}\nabla f^{(\ell)}(\bx^{[m]})\right),
\end{equation}
where $\alpha^{[m]}$ is the safeguarded BB steplength, defined as~\cite{andreani2007on}
\begin{equation}
	\label{eq:BB steplength}
	\alpha^{[m]}=\left\{
	\begin{aligned}
		&\min\left\{\max\left\{\frac{s^{\T}s}{s^{\T}y}, \alpha_{\min}\right\}, \alpha_{\max}\right\},~\text{if}~s^{\T}y>0,\\
		&\alpha_{\max},~\text{otherwise}.
	\end{aligned}
	\right.
\end{equation}
with $s=\bx^{[m]}-\bx^{[m-1]}, y=\nabla f^{(\ell)}(\bx^{[m]}) - \nabla f^{(\ell)}(\bx^{[m-1]})$, and $\delta$ is initialized to 1 and successively reduced by a backtracking procedure until Grippo--Lampariello--Lucidi (GLL) line search criterion \cite{grippo1986nonmonotone}
\begin{equation}
	\label{eq:GLL}
	\begin{aligned}
	 	f^{(\ell)}(\bx^{[m+1]})  \leq&\max_{0\leq q\leq \min\{m, Q-1\}} f^{(\ell)}(\bx^{[m-q]})\\
		& +\delta_1\left(\bx^{[m+1]}-\bx^{[m]}\right)^{\T}  \nabla f^{(\ell)}(\bx^{[m]})
	\end{aligned}
\end{equation}
is satisfied, where $Q \in \mathbb{N}$, $0<\delta_1<1$. The nonmonotone SPG method terminates if an  $\epsilon$-stationary point of \eqref{eq:ALM subproblem} is obtained, which is characterized as
\begin{equation}
	\label{eq:SPG termination}
	\left\|\mathcal{P}_{[-\eta, \eta]^{2N_t}}\left(\bx^{[m]}-\nabla f^{(\ell)}(\bx^{[m]})\right) - \bx^{[m]}\right\|\leq \epsilon_2,
\end{equation}
where $\epsilon_2>0$ is a prescribed threshold. The nonmonotone SPG method is summarized in Algorithm \ref{alg:SPG}.

\begin{algorithm}[t]
\caption{Nonmonotone SPG Method for Solving \eqref{eq:ALM subproblem}}
\label{alg:SPG}
\begin{algorithmic}[1]
\State \textbf{Input}: $\alpha_{\max}>\alpha_{\min}>0, 0< \delta_1, \delta_2 < 1$, $Q \in \mathbb{N}$,   $\epsilon_2>0$;
\State \textbf{Initialization}: $\bx^{[0]}\in [-\eta, \eta]^{2N_t}, \alpha^{[0]}\in~[\alpha_{\min}, \alpha_{\max}]$, $m\gets 0$;
\Repeat
\State Initialize $\alpha^{[m]}$ according to \eqref{eq:BB steplength}, $\delta \gets 1$;
\Repeat
\State $\bx^{[m+1]}\gets \mathcal{P}_{[-\eta, \eta]^{2N_t}}\left(\bx^{[m]}-\delta \alpha^{[m]}\nabla f^{(\ell)}(\bx^{[m]})\right)$;
\State $\delta \gets \delta/2$;
\Until{GLL line search criterion \eqref{eq:GLL} is satisfied};
\State $m \gets m+1$;
\Until a stationary point of \eqref{eq:ALM subproblem} is obtained.
\end{algorithmic}
\end{algorithm}

\textit{Step 2 (Update multipliers)} We update
\begin{equation}
	\label{eq:uv update}
	\hspace{-0.6em}
	\begin{aligned}
		\bu^{(\ell+1)}&=\min\{\max\{ \bu^{(\ell)}-\rho_u^{(\ell)}(\bC\bx^{(\ell)}-\boldsymbol{\gamma}), \mathbf{0}\}, \bu_{\max}\}, \\
		\bv^{(\ell+1)}&=\min\{\max\{ \bv^{(\ell)}+\rho_v^{(\ell)}(\bD\bx^{(\ell)}-\boldsymbol{\zeta}), \mathbf{0}\}, \bv_{\max}\},
	\end{aligned}
\end{equation}
where $\bu_{\max}, \bv_{\max}$ serve as safeguards to guarantee algorithm convergence \cite{galvan2020alternating}.

\textit{Step 3 (Update penalty parameters)}
The ALM penalty parameters are updated as in \eqref{eq:rho_u rho_v update}, as shown at the top of this page, where
\begin{equation*}
	\begin{aligned}
		\sigma_1(\bx, \bu, \rho_u) &=\left\|\max\left\{ -\bC\bx+\boldsymbol{\gamma},-\frac{\bu}{\rho_u}\right\}\right\|_\infty \\
		\sigma_2(\bx, \bv, \rho_v) &=\left\|\max\left\{ \bD\bx-\boldsymbol{\zeta},-\frac{\bv}{\rho_v}\right\}\right\|_\infty
	\end{aligned}
\end{equation*}
characterize the level of infeasibility and violation of complementarity with respect to constraint \eqref{eq:penalty problem comm constraint} and \eqref{eq:penalty problem amplitude constraint}, respectively. 
The overall algorithm terminates when  
\[
\sigma(\bx^{(\ell+1)}, \bu^{(\ell+1)}, \bv^{(\ell+1)}, \rho_u^{(\ell+1)}, \rho_v^{(\ell+1)}) \leq \epsilon_1,
\] where $\sigma(\bx, \bu, \bv, \rho_u, \rho_v) = \max\{\sigma_1(\bx, \bu, \rho_u), \sigma_2(\bx, \bv, \rho_v)\}
$ and $\epsilon_1>0$ is a prescribed threshold.

The detailed implementation of ALM is listed in Algorithm~\ref{alg:ALM}. The convergence of Algorithm~\ref{alg:ALM} to a stationary point of~\eqref{eq:penalty problem} is guaranteed under the conditions in~\cite[Section~4]{andreani2007on}.

\begin{algorithm}[t]
\caption{ALM for Solving \eqref{eq:penalty problem}}
\label{alg:ALM}
\begin{algorithmic}[1]
\State \textbf{Input}: $\bu_{\max}, \bv_{\max} \geq \mathbf{0}, \tau \in (0,1), \iota>1$,  $\epsilon_1>0$;
\State \textbf{Initialization}: $\bx^{(0)}\in [-\eta, \eta]^{2N_t}, \bu^{(0)}\in [\mathbf{0}, \bu_{\max}], \bv^{(0)} \in [\mathbf{0}, \bv_{\max}], \rho_u^{(0)}, \rho_v^{(0)}>0, \ell\gets 0$;
\While{$\sigma(\bx^{(\ell)}, \bu^{(\ell)}, \bv^{(\ell)}, \rho_u^{(\ell)}, \rho_v^{(\ell)})\geq \epsilon_1$}
\State Obtain $\bx^{(\ell+1)}$ by solving  \eqref{eq:ALM subproblem} via Algorithm \ref{alg:SPG};
\State Update Lagrange multipliers $\bu, \bv$ using \eqref{eq:uv update};
\State Update the ALM penalty parameters $\rho_u, \rho_v$ using \eqref{eq:rho_u rho_v update};
\State $\ell \gets \ell+1$;
\EndWhile
\end{algorithmic}
\end{algorithm}

\subsection{Solution Refinement}
\label{subsection:refinement}
Once problem~\eqref{eq:penalty problem} is solved, two cases may arise: either the elements of the obtained solution are already binary, or some of them are not. 
For the former case, the algorithm terminates since the solution is already feasible for problem~\eqref{eq:original problem}. 
For the latter case, we need to refine the solution to get a feasible binary vector. To this end, we employ two strategies, namely, a local search procedure and a cutting-plane strategy.

\paragraph{Local Search Procedure}
As established in Theorem~\ref{theorem: Ix lower bound}, when the ALM algorithm converges to a local minimizer $\hat{\bx}$ of problem \eqref{eq:penalty problem}, a large proportion of the solution are already binary. Leveraging this property, we perform a local search around the obtained solution to identify a nearby feasible binary vector over the neighborhood set
\begin{equation*}
	\mathcal{N}(\hat{\bx}, c)
	:= \left\{ \bx \in \{\eta, -\eta\}^{2N_t} \,\middle|\, \|\bx - \bar{\bx}\|_0 \leq c \right\},
\end{equation*}
where $\bar{\bx} = \eta \cdot \sign(\hat{\bx})$, and $c \in \mathbb{N}$ controls the size of the neighborhood. The cardinality of $\mathcal{N}(\hat{\bx}, c)$ is given by $\lvert \mathcal{N}(\hat{\bx}, c) \rvert = \binom{2N_t}{c}$. We search over $\mathcal{N}(\hat{\bx}, c)$ to identify vectors satisfying constraint~\eqref{eq:original problem comm constraint}, and select the one that maximizes
\begin{equation*}
	f_{\textnormal{1-bit}}(\bx)=\frac{2}{\sigma_r^2} \langle \boldsymbol{\kappa}(\bx), \mathbf{f}(\bx) \rangle.
\end{equation*}
\paragraph{Cutting-Plane Strategy}
If there exists no feasible point within $\mathcal{N}(\hat{\bx}, c)$ that satisfies~\eqref{eq:original problem comm constraint}, we further introduce a cutting-plane strategy to eliminate the current local minimizer and its surrounding region from the feasible set of problem~\eqref{eq:penalty problem} because it contains no feasible solution. 
In fact, the region $\mathcal{N}(\hat{\bx}, c)$ contains all vectors $\bx \in  \{-\eta, \eta\}^{2N_t}$ that differs from $\bar{\bx}$ in at most $c$ positions. Hence, the inner product $\bar{\bx}^{\T} \bx$ is lower bounded by $(2N_t-2c)\eta^2$, and vice versa. Hence, for all~$\bx \in \{\eta, -\eta\}^{2N_t}$, we have
\begin{equation*}
	\bx \in \mathcal{N}(\hat{\bx}, c)
	\Longleftrightarrow
	\bar{\bx}^{\T} \bx \geq (2N_t-2c)\eta^2,
\end{equation*}
and we can deduce that 
\begin{equation*}
	\bx \notin \mathcal{N}(\hat{\bx}, c)
	\Longleftrightarrow
	\bar{\bx}^{\T} \bx \leq (2N_t-2c-2)\eta^2.
\end{equation*}
Therefore, we add the following linear inequality constraint to eliminate $\mathcal{N}(\hat{\bx}, c)$ and arrive at a refined feasible region constraint for problem~\eqref{eq:original problem}:
\begin{equation}
	\label{eq:cut}
	\bar{\bx}^{\T} \bx \leq (2N_t-2c-2)\eta^2.
\end{equation}
We incorporate \eqref{eq:cut} into problem~\eqref{eq:original problem} and perform Algorithm~\ref{alg:ALM} again, initialized with $\hat{\bx}$ as the starting point.

\begin{algorithm}[t]
\caption{Algorithm for Solving \eqref{eq:original problem}}
\label{alg:overall}
\begin{algorithmic}[1]
\State \textbf{Input:} $c\in \mathbb{N}$;
\Repeat
    \State Solve problem \eqref{eq:penalty problem} via Algorithm~\ref{alg:ALM} to obtain $\hat{\bx}$;
    
    \If{$\hat{\bx} \in \{-\eta,\eta\}^{2N_t}$}
        \State \Return $\hat{\bx}$;
    \EndIf
    
    \State $\mathcal{F} \gets \{ \bx \in \mathcal{N}(\bar{\bx}, c) \mid \bC\bx \geq \boldsymbol{\gamma} \}$;
    
    \If{$\mathcal{F} \neq \varnothing$}
        \State \Return $\argmin_{\bx \in \mathcal{F}} f_{\text{1-bit}}(\bx)$;
    \EndIf
    
    \State Update the current best solution $\bx^\star$;
    \State Add constraint \eqref{eq:cut} to problem \eqref{eq:penalty problem}, return to step 3;

\Until{a stopping criterion is satisfied}.
\end{algorithmic}
\end{algorithm}

To summarize, the overall algorithm including ALM  and the refinement step is shown in Algorithm~\ref{alg:overall}, where in line 12 the current best solution $\bx^*$ is determined by finding
\begin{equation*}
	\mathop{\arg\max}_{\bx \in \mathcal{N}(\hat{\bx}, c)}~ f_{\text{1-bit}}(\bx)+\xi\min_{k}\{\bc_k^{\T}\bx\},
\end{equation*}
where $\bc_k^{\T}$ is the $k$-th row of $\bC$, $\xi$ acts as a balancing parameter between radar and communication performance.

\subsection{Complexity Analysis}

We analyze the big-O computational complexity of the proposed algorithm in terms of the number of transmit antennas $N_t$, the number of users $K$, the total number of SPG iterations $N_{\text{iter}}$ and the number of cutting planes $N_{\text{cut}}$. The computational cost of evaluating the gradient $\nabla L(\bx,\bu,\bv,\rho_u,\rho_v)$ is in the order of $\mathcal{O}(N_t^2 + K N_t)$. Updating the dual variables and penalty parameters requires $\mathcal{O}(K)$ operations per iteration. Consequently, the per-iteration complexity of solving the penalized problem is $\mathcal{O}(N_t^2 + K N_t)$. For the local search procedure, the cardinality of the search set $\mathcal{N}(\hat{\bx}, c)$ scales as $\mathcal{O}(N_t^2)$ in our implementation with $c=2$. For each vector $\bx \in \{-\eta, \eta\}^{2N_t}$, verifying whether $\bx$ satisfies constraint~\eqref{eq:original problem comm constraint} incurs a computational cost of $\mathcal{O}(K)$. Therefore, the overall computational complexity of the local search step is $\mathcal{O}(K N_t^2)$.

By summing up the above components, the overall computational complexity of the proposed algorithm is given by $\mathcal{O}\left(N_{\text{iter}} (N_t^2 + K N_t) + N_{\text{cut}}K N_t^2 \right).$

\section{Numerical Results} 
\label{section:numerical results}
In this section, we present simulation results to validate the theoretical analysis and evaluate the performance of the proposed design. Specifically, we first verify Theorem~\ref{theorem: Ix lower bound} and examine the convergence behavior of Algorithm~\ref{alg:ALM}. 
We then compare the proposed design with state-of-the-art (SOTA) MIMO DFRC system designs. 
Finally, we evaluate the performance of the proposed design under different system settings.

The simulation settings are as follows. Unless stated otherwise, the number of transmit antennas and receive antennas are set to $N_t=32$ and $N_r=32$, respectively, and the number of users is $K=6$.  The transmit power is fixed at $P=20$~dBm. The communication SNR, defined as $\textnormal{SNR}_{\text{comm}} = P/\sigma_c^2$, is set to $20$ dB, while the radar SNR is fixed at $10$ dB. We adopt 8-PSK modulation and impose a SEP requirement of $10^{-4}$. The communication channel $\mathbf{H}$ is assumed to undergo Rayleigh fading, with independent and identically distributed entries drawn from $\mathcal{CN}(0,1)$. The azimuth angle of target $\theta$ is drawn from a uniform distribution on $[-\pi/3, \pi/3]$. All results are averaged over 100 independent channel realizations.
 
 The parameters used in Algorithm~1 and Algorithm~2 are set as follows: $\mathbf{u}_{\max}=\mathbf{v}_{\max}=10^{4}\cdot\mathbf{1}, \tau=0.8, \iota=3, \alpha_{\max}=10^{30}, \alpha_{\min}=10^{-30}, \delta_1 =10^{-4}, \epsilon_1=\epsilon_2=10^{-4}$, $Q=5$, $c=2$. The initialization is set as follows: $\bx^{(0)}=\mathbf{0}$, $\mathbf{u}^{(0)}=\mathbf{0}$, $\mathbf{v}^{(0)}=\mathbf{0}$, $\rho_u^{(0)}=0.1$, $\rho_v^{(0)}=0.01$. Algorithm~\ref{alg:overall} is terminated if the number of cutting planes $N_{\text{cut}}$ exceeds 3.

Define the maximum likelihood (ML) estimator of $\theta$ as $
	\hat{\theta}^{\text{ML}}= \mathop{\arg\max}_{\theta}~ l_{\text{1-bit}}(\theta|\bq_1,\bq_2,\ldots,\bq_T),$
where the maximizer is found via an exhaustive grid search over $[-\pi/2,\pi/2]$. We adopt root-CRB (RCRB) and root-MSE (RMSE) of the ML estimator as the sensing performance metrics, given by
 \[
 \mathrm{RCRB} := \sqrt{\frac{1}{F_{\text{1-bit}}(\bX)}}, \quad \mathrm{RMSE} := \sqrt{ \mathbb{E}\left[ \left(\hat{\theta}^{\text{ML}}-\theta\right)^2\right] },
 \]
 and adopt symbol error rate (SER) as the communication performance metric.

\subsection{Verification of Theorem~\ref{theorem: Ix lower bound}}
\label{subsection:verification of theorem ix lower bound}

We first verify the upper bound on the number of non-binary elements, $|I(\hat{\bx})|$, established in Theorem~\ref{theorem: Ix lower bound} by examining the solutions obtained via Algorithm~\ref{alg:ALM}.
Specifically, we consider four scenarios with different communication SNRs, modulation orders, and SEP requirements, expressed by triplets:
\begin{enumerate}
	\centering
	\item[]\hspace{-3em} \textbf{scenario 1:} \quad (25 dB, 4-PSK, $10^{-3}$),
	\item[]\hspace{-3em} \textbf{scenario 2:} \quad (25 dB, 4-PSK, $10^{-4}$),
	\item[]\hspace{-3em} \textbf{scenario 3:} \quad (25 dB, 8-PSK, $10^{-4}$),
	\item[]\hspace{-3em} \textbf{scenario 4:} \quad (20 dB, 8-PSK, $10^{-4}$).
\end{enumerate}
These scenarios are ordered from ``easy'' to ``hard'' in the sense that the constraint in~\eqref{eq:original problem comm constraint} becomes progressively more stringent from scenario~1 to scenario~4.

\begin{figure}[t]
	\centering 
	\includegraphics[width=0.35\textwidth]{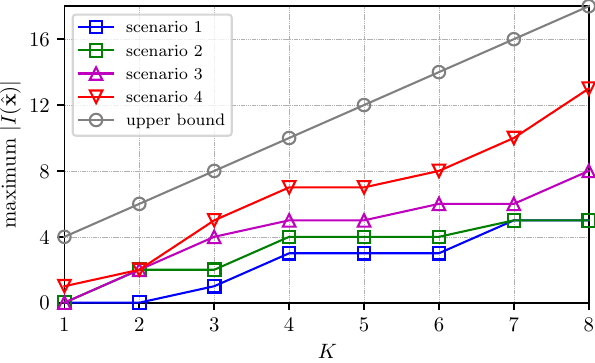}
	\caption{Maximum $|I(\hat{\bx})|$ versus the number of users $K$.}
	\label{fig:verification of lemma}
\end{figure}

We  examine the cardinality of $I(\hat{\bx})$ as a function of the number of users $K$ for the considered scenarios. For each value of $K$, we generate $1,000$ independent problem instances and plot the worst-case maximum values of $|I(\hat{\bx})|$ among these instances.
The results are shown in Fig.~\ref{fig:verification of lemma}.
It can be seen that the simulated maximum value of  $|I(\hat{\bx})|$ is below the theoretical upper bound in Theorem~\ref{theorem: Ix lower bound}, which provides numerical evidence for Theorem~\ref{theorem: Ix lower bound}.

\begin{figure}[t]
	\centering 
	\includegraphics[width=0.35\textwidth]{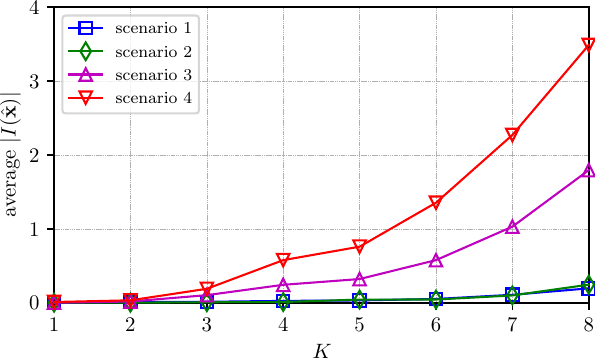}
	\caption{Average $|I(\hat{\bx})|$ versus the number of users $K$.}
	\label{fig:average Ix}
\end{figure}

We also plot the average values of $|I(\hat{\bx})|$ versus $K$ under different scenarios in Fig.~\ref{fig:average Ix}. 
By comparing Figs. \ref{fig:verification of lemma} and \ref{fig:average Ix}, it is seen that the averaged number of non-binary elements $|I(\hat{\bx})|$ is much smaller than the worst case.  
In particular, when $K$ is small to moderate, the average values for Scenarios~1-4 are close to zero, indicating that most solution elements have already fully binary entries. 
This also relieves the computational burden in the final refinement step. 

\subsection{Convergence Behavior of Algorithm~\ref{alg:ALM}}
In this subsection, we investigate the convergence behavior of Algorithm~\ref{alg:ALM}. Fig.~\ref{fig:convergence_behavior}(a) depicts the change of $\sigma(\bx, \bu, \bv, \rho_u, \rho_v)$ with respect to the outer iterations, while Fig.~\ref{fig:convergence_behavior}(b) shows the number of SPG steps required to reach a stationary point of each ALM subproblem versus the outer iteration index.
It is observed that Algorithm~\ref{alg:ALM} converges within   a dozen outer iterations. Moreover, the computational cost associated with solving each ALM subproblem reduces as the algorithm proceeds, and remains low except for the first several outer iterations.
\begin{figure}[t]
	\centering
	\subfigure[\hspace{-3em}]{\includegraphics[height=2.3in]{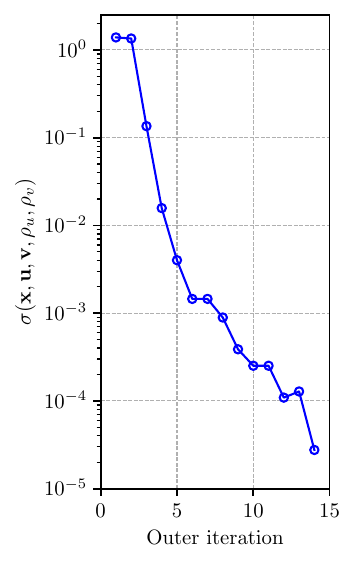}} \hspace{1em}
	\subfigure[\hspace{-3em}]{\includegraphics[height=2.3in]{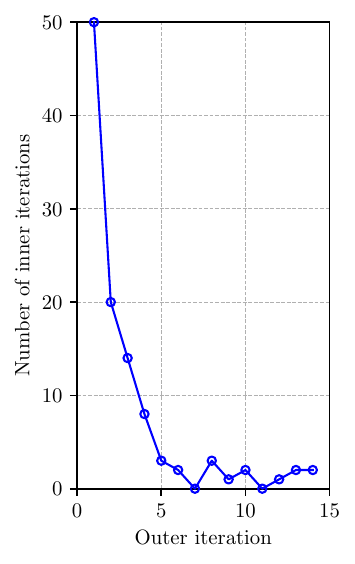}}
	\caption{Convergence behavior of Algorithm~\ref{alg:ALM}.}
	\label{fig:convergence_behavior}
\end{figure}

\subsection{Comparison with SOTA MIMO DFRC System Designs} \label{subsection:benchmark}
In this subsection, we compare our design with those tailored for MIMO DFRC systems with 1-bit DACs and high-resolution ADCs. Specifically, we consider the following benchmark schemes.
\begin{itemize}
	\item {\bf Proposed without amplitude constraints}: MIMO DFRC design with problem formulation \eqref{eq:formulation unquantized}, which differs from \eqref{eq:formulationun_ac} only in that the amplitude constraint \eqref{eq:amplitude constraint reduced} is not involved. The problem is solved using Algorithm~\ref{alg:overall}.
	\item {\bf ALM--BSUM} \cite{wu2025quantized}: Minimizing the MSE between the designed and desired beampatterns, subject to the CI-based communication QoS constraint \eqref{eq:original problem comm constraint} and binary constraint \eqref{eq:original problem amplitude constraint}, where the desired beampattern is set to
		\begin{equation}
			d(\theta)= \left\{ \begin{aligned} 
				1, & \text{~if~} \theta \in \left[\theta_0-\frac{\Delta_{\theta}}{2}, \theta_0+\Delta_{\theta}\right];  \\
				0, & \text{~otherwise},
			\end{aligned} \right.
		\end{equation}
		with $\Delta_{\theta}=10^\circ$ being the desired beamwidth. The problem is solve by the ALM, with ALM subproblems solved by a custom-built block successive upper-bound minimization (BSUM) algorithm. 
	
	\item {\bf CEO} \cite{wang2025interference}: Maximizing a weighted sum of the minimum multiuser CI scaling factors for communication and the illumination power at direction $\theta$ for radar sensing, subject to binary constraint \eqref{eq:original problem binary constraint}, where the weighting factor is carefully chosen to meet SEP requirement. The problem is solved using a cross-entropy optimization framework.	
\end{itemize}

\begin{figure}[t]
	\centering
	\includegraphics[width=0.46\textwidth]{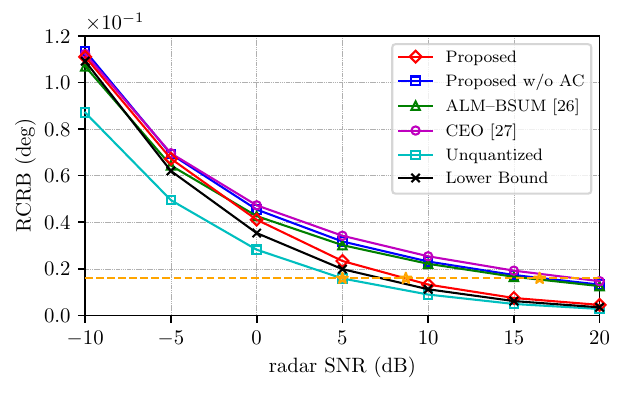}
	\caption{RCRB versus radar SNR.}
	\label{fig:RCRB_vs_radarSNR}
\end{figure}
 
\begin{figure}[t]
	\centering
	\includegraphics[width=0.46\textwidth]{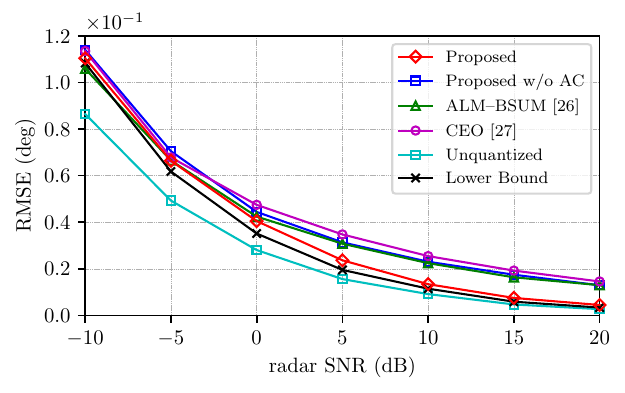}
	\caption{RMSE versus radar SNR.}
	\label{fig:RMSE_vs_radarSNR}
\end{figure}

In Fig.~\ref{fig:RCRB_vs_radarSNR} and Fig.~\ref{fig:RMSE_vs_radarSNR}, we compare the achieved RCRB and RMSE under different radar SNR levels. 
The ``Unquantized'' curve refers to the unquantized CRB with the transmit signal matrix obtained by solving problem~\eqref{eq:formulation unquantized} using Algorithm~\ref{alg:overall}, and the ``Lower Bound'' curve refers to the lower bound
\[
{\textnormal{RCRB}}_{\textnormal{1-bit}} \geq \sqrt{\frac{\pi}{2}} {\textnormal{RCRB}}_{\textnormal{unquant.}},
\]
which is obtained by inverting both sides of \eqref{eq:1-bit FI upper bound} and then taking the square root.

The performance of proposed method approaches the lower bound across all the tested radar SNR regimes. In addition, the proposed method outperforms the benchmark schemes significantly in the medium to high SNR regimes. In particular, as illustrated by the dashed orange line in Fig.~\ref{fig:RCRB_vs_radarSNR}, the proposed approach operating at a radar SNR of 8~dB can attain a RCRB comparable to that of a system equipped with high-resolution ADCs operating at 5~dB, and surpass the benchmark schemes operating at 16~dB under the same conditions. It is worth noting that the proposed design significantly outperforms its counterpart without the amplitude constraint, highlighting the effectiveness of the constraint in \eqref{eq:original problem amplitude constraint}.

\begin{figure}[t]
	\centering 
	\subfigure[RCRB versus SER.]{\includegraphics[height = 2.3in]{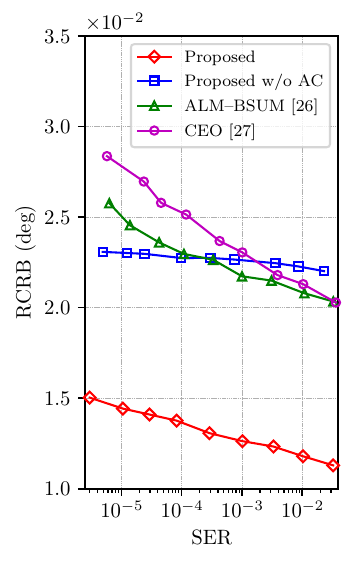}} \hspace{1.2em}
	\subfigure[RMSE versus SER.]{\includegraphics[height = 2.3in]{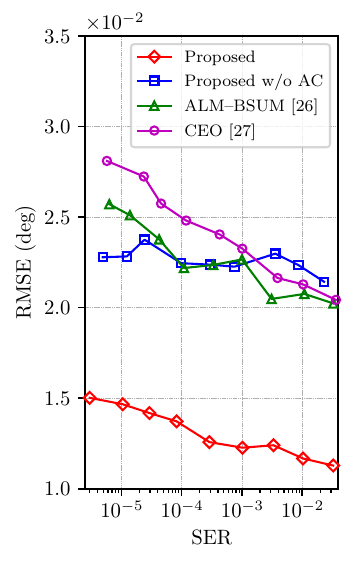}}
	\caption{Achieved radar and communication performance tradeoffs of different approches.}
	\label{fig:RCRB_and_RMSE_vs_SER}
\end{figure}

In Fig.~\ref{fig:RCRB_and_RMSE_vs_SER}, we illustrate the tradeoff between radar and communication performance for the considered approaches. 
It is seen that the proposed approach achieves lower RCRB and RMSE than the benchmark schemes under the same communication SER requirements.

\begin{table}[t]
\centering
\caption{Average running time (second) versus the number of transmit antennas $N_t$.}
\label{tab:runtime_vs_Nt}
\begin{tabular}{ccccccc}
\toprule
\multirow{3}{*}{\textbf{Algorithm}} & \multicolumn{6}{c}{$N_t$} \\
\cmidrule(lr){2-7}
 & 8 & 16 & 24 & 32 & 64 & 128 \\
\midrule
Proposed  & 0.02 & 0.06 & 0.10 & 0.18 & 0.57 & 1.86 \\
ALM-BUSM~\cite{wu2025quantized} & 0.44 & 1.22 & 2.72 & 3.11 & 4.93 & 6.21 \\
CEO~\cite{wang2025interference} & 0.11 & 0.56 & 1.38 & 3.09 & 15.82 & 79.73 \\
\bottomrule
\end{tabular}
\end{table}
 
In Table~\ref{tab:runtime_vs_Nt}, we compare the average running time of the proposed algorithm with those adopted in \cite{wu2025quantized} and \cite{wang2025interference}. The MIMO DFRC system is configured with $N_r = N_t$, $K = N_t/8$. All algorithms are implemented in Python and executed on a macOS Tahoe system with an Apple M2 Pro chip with 16 GB of memory. The proposed algorithm achieves significantly lower runtime than both benchmark algorithms, especially for large $N_t$, indicating its lower computational complexity.

\subsection{Different System Settings}

\begin{figure}[t]
	\centering
	\includegraphics[width=0.4\textwidth]{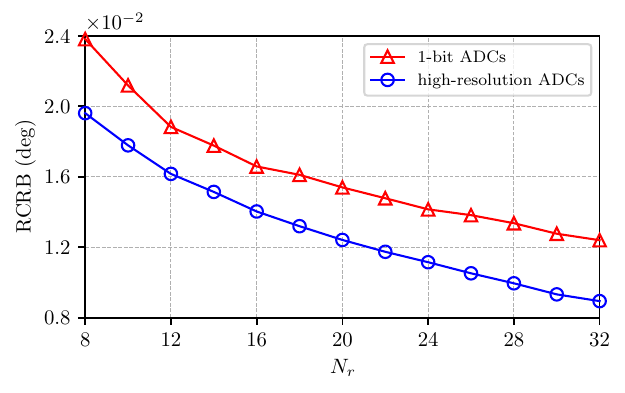}
	\caption{RCRB versus the number of receive antennas $N_r$.}
	\label{fig:RCRB_vs_Nr}
\end{figure}
\begin{figure}[t]
	\centering
	\subfigure[RCRB versus $K$.]{\includegraphics[height=2.3in]{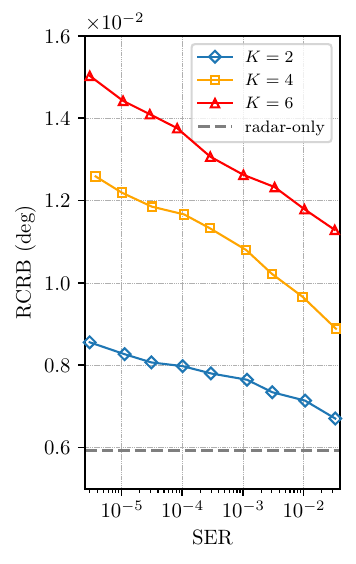}}
	\label{fig:RCRB_vs_SER}\hspace{0.2em}
	\subfigure[RMSE versus $K$.]{\includegraphics[height=2.3in]{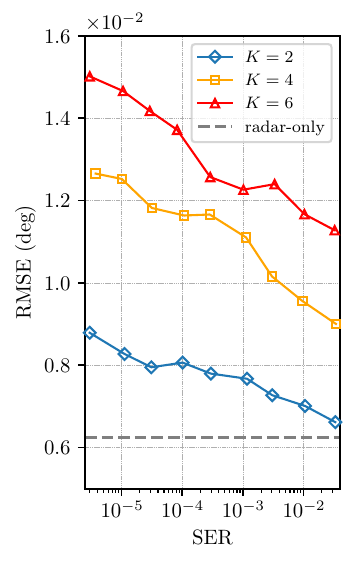}}
	\label{fig:RMSE_vs_SER}
	\caption{Achieved radar and communication performance tradeoffs for different $K$.}
	\label{fig:RCRB_and_RMSE_vs_K}
\end{figure}

In Fig.~\ref{fig:RCRB_vs_Nr}, we compare the achieved RCRB of a massive MIMO DFRC system equipped with 1-bit ADCs to that of a counterpart employing high-resolution ADCs, as a function of the number of receive antennas $N_r$. The results indicate that the system with 1-bit ADCs can attain sensing performance comparable to that of system with high-resolution ADCs, provided that 1.6 times as many receive antennas are deployed.

In Fig.~\ref{fig:RCRB_and_RMSE_vs_K}, we illustrate the tradeoff between RCRB/RMSE and SER for different numbers of users $K$. The dashed gray line is obtained by solving problem~\eqref{eq:original problem} without communication QoS constraint~\eqref{eq:original problem comm constraint}. The results show that, as $K$ increases, the communication constraint in~\eqref{eq:original problem comm constraint} becomes more restrictive, leading to a higher achieved RCRB.

\section{Conclusion}
\label{section:conclusion}
In this paper, we have investigated the DFRC design for a massive MIMO system equipped with both 1-bit DACs and ADCs. Our formulation optimizes the 1-bit transmit signals to minimize the 1-bit CRB while satisfying the communication QoS requirements. 
By exploiting the diminishing property of the 1-bit Fisher information function, we eliminated regions where the objective function is clearly suboptimal through the introduction of the amplitude constraints. We have proposed an efficient method, involving the penalty method, the ALM algorithm,  the SPG algorithm, and a dedicated refining step, to tackle the resulting nonconvex design problem. 
Simulation results verified our analytical findings and demonstrated that the proposed approach exhibits competitive performance compared to benchmark schemes.

\appendices
\section{Proof of proposition \ref{proposition:asymptotic 1-bit FI}}
\label{section:proof of proposition asymptotic}
\begin{proof}
	We first show that
	$\displaystyle \lim_{|x|\rightarrow +\infty}{\rho(x)x^2=0}$. Since $\rho(x)x^2$ is even, it suffices to consider the limit as $x\rightarrow +\infty$. We have
	\begin{equation}
		\begin{aligned}
			\lim_{x\rightarrow +\infty}{\rho(x)x^2}&= \lim_{x\rightarrow +\infty}\frac{x^2\phi^2(x)}{\Phi(x)[1-\Phi(x)]}\\
			&\overset{(a)}{=} \lim_{x\rightarrow +\infty}\frac{x^2\phi^2(x)}{1-\Phi(x)}\\
			&\overset{(b)}{=} \lim_{x\rightarrow +\infty}\frac{2x\phi^2(x)+2x^2\phi(x)\phi'(x)}{-\phi(x)} \\
			&= \lim_{x\rightarrow +\infty}-[2x\phi(x)+2x^2\phi'(x)] \\
			&= \lim_{x\rightarrow +\infty}\sqrt{\frac{2}{\pi}}(x^3-x)e^{-\frac{1}{2}x^2}=0, \\
		\end{aligned}
	\end{equation}
	where $(a)$ is due to the fact that $\lim_{x\rightarrow +\infty} \Phi(x)=1$ and $(b)$ follows from L'H\^opital's rule. 
	
	Next, we prove Proposition~\ref{proposition:asymptotic 1-bit FI}. For all $t \in [T]$ and $n \in [N_r]$, note that 
	$$\Re(r_{t,n}) = \mathbf{b}_{n}^{\T}\mathbf{x}_t, \quad \Im(r_{t,n}) = \mathbf{b}_{n+N_r}^{\T}\mathbf{x}_t,$$
	$$\frac{\partial \Re(r_{t,n})}{\partial \theta} = \mathbf{a}_{n}^{\T}\mathbf{x}_t, \quad \frac{\partial \Im(r_{t,n})}{\partial \theta} = \mathbf{a}_{n+N_r}^{\T}\mathbf{x}_t,$$
	where $\mathbf{a}_{n}^{\T}$ is the $n$-th row of the matrix $\bA$, which is defined in~\eqref{eq:matrix A}, and $\mathbf{b}_{n}^{\T}$ is the $n$-th row of the matrix $\mathbf{B}$, which is defined as
	$$
		\mathbf{B}= \begin{bmatrix}
			\Re(\tilde{\mathbf{B}}) & -\Im(\tilde{\mathbf{B}}) \\
			\Im(\tilde{\mathbf{B}}) & \Re(\tilde{\mathbf{B}})
		\end{bmatrix},
	$$
	with $\tilde{\mathbf{B}}=\beta \ba_r\ba_t^\H$. Then the 1-bit Fisher information can be expressed by
	\begin{equation}
		\begin{aligned}
		F_{\textnormal{1-bit}}(\bX)=& \sum_{t=1}^T \sum_{n=1}^{2N_r} \Bigg[\rho\left(\frac{\sqrt{2}}{\sigma_r}\mathbf{b}_n^{\T}\bx\right)\left(\frac{\sqrt{2}}{\sigma_r}\mathbf{a}_n^\T \bx\right)^2 \Bigg].
		\end{aligned}
	\end{equation}
	Hence, it suffices to show that for all $t \in [T]$ and for all $ n \in [2N_r]$, if $\mathbf{b}_n^{\T}\bx \neq 0$, it holds that
	$$\lim_{\text{SNR}_{\text{radar}}\rightarrow +\infty} \rho\left(\frac{\sqrt{2}}{\sigma_r}\mathbf{b}_n^{\T}\bx\right)\left(\frac{\sqrt{2}}{\sigma_r}\mathbf{a}_n^\T \bx\right)^2=0,$$
	which is a direct result of $\displaystyle \lim_{|x|\rightarrow +\infty}{\rho(x)x^2=0}$. This completes the proof.
\end{proof}

\section{Proof of Proposition \ref{proposition:reduction of constraints}}
\label{section:proof of proposition reduction}
\begin{proof}
	Since $\mathcal{A}(\zeta) \subseteq \mathcal{B}(\zeta)$ is trivial, it suffices to show that $\mathcal{B}(\zeta) \subseteq \mathcal{A}(\sqrt{2} \zeta)$. For all $\bx_t \in \mathcal{B}(\zeta)$, we have
	$$
		A_1(\bx_t) = \max\left\{ \left|\frac{\sqrt{2}}{\sigma_r}\Re(r_{t,1})\right|, \left|\frac{\sqrt{2}}{\sigma_r}\Im(r_{t,1})\right| \right\} \leq \zeta,
	$$
	which implies 
	$$ 
		\left|\frac{\sqrt{2}}{\sigma_r}r_{t,1}\right|=\sqrt{\left|\frac{\sqrt{2}}{\sigma_r}\Re(r_{t,1})\right|^2+\left|\frac{\sqrt{2}}{\sigma_r}\Im(r_{t,1})\right|^2} \leq \sqrt{2}\zeta.
	$$
	Since $\mathbf{r}_t=\beta\ba_r(\theta)\ba_t^{\H}(\theta)\tilde{\bx}_t$, we have
	$$
		|r_{t,n}|=|\beta||\ba_t^{\H}(\theta)\tilde{\bx}_t|, ~\forall~n \in [N_r],
	$$
	and hence 
	$$
		\left|\frac{\sqrt{2}}{\sigma_r}r_{t,n}\right| \leq \sqrt{2}\zeta,~ \forall~n \in [N_r],
	$$
	which further implies 
	$$
		\max\left\{ \left|\frac{\sqrt{2}}{\sigma_r}\Re(r_{t,n})\right|, \left|\frac{\sqrt{2}}{\sigma_r}\Im(r_{t,n})\right| \right\} \leq \sqrt{2}\zeta, \forall~n \in [N_r],
	$$
	i.e., $\bx_t \in \mathcal{A}(\sqrt{2} \zeta)$. The proof is completed.
\end{proof}

\section{Proof of Proposition \ref{proposition:well approximation}}
\label{section:proof of proposition well approximation}

\begin{proof}
	The upper bound in \eqref{eq:obj approximation} follows directly from \eqref{eq:1-bit FI upper bound}. The lower bound is derived as follows.
	
	From Proposition~\ref{proposition:reduction of constraints}, for any $\bX$ satisfying \eqref{eq:amplitude constraint reduced}, it holds that
	$$
		A_n(\bx_t) \leq \sqrt{2}\zeta, \forall~t \in [T], n \in [N_r].
	$$
	
	Furthermore, by the monotonicity of $\rho(\cdot)$ on $(0,\infty)$, we can deduce that for all $t \in [T] $ and for all $n \in [N_r]$,
\begin{equation*}
	\min\left\{\rho \left(\frac{\sqrt{2}}{\sigma_r}\Re(r_{t,n})\right), \rho \left(\frac{\sqrt{2}}{\sigma_r}\Im(r_{t,n})\right)\right\}\geq \rho(\zeta).
\end{equation*}
Hence, the 1-bit Fisher information is lower bounded as
\begin{equation*}
	\begin{aligned}
	&F_{\textnormal{1-bit}}(\bX)\\
	&=\frac{2}{\sigma_r^2} \sum_{t=1}^T \sum_{n=1}^{N_r} \Bigg[\rho\left(\frac{\sqrt{2}}{\sigma_r}\Re(r_{t,n})\right)\left(\frac{\partial \Re(r_{t,n})}{\partial \theta}\right)^2 \\
	&\hspace{1em} + \rho\left(\frac{\sqrt{2}}{\sigma_r}\Im(r_{t,n})\right) \left(\frac{\partial \Im(r_{t,n})}{\partial \theta}\right)^2 \Bigg]\\
	&\geq \frac{2}{\sigma_r^2} \rho(\sqrt{2}\zeta)\sum_{t=1}^T \sum_{n=1}^{N_r} \Bigg[\left(\frac{\partial \Re(r_{t,n})}{\partial \theta}\right)^2 +\left(\frac{\partial \Im(r_{t,n})}{\partial \theta}\right)^2 \Bigg]\\
	&=\rho(\sqrt{2}\zeta)F_{\textnormal{unquantized}}(\bX).
	\end{aligned} 
\end{equation*}
The proof is completed.
\end{proof}

\section{Proof of Theorem \ref{theorem: Ix lower bound}}
\label{section:proof of theorem Ix lower bound}

The proof of Theorem \ref{theorem: Ix lower bound} relies on the following lemma:
\begin{lemma}
	\label{lemma:extreme point}
	Consider a full-dimensional polytope $\mathcal{F}=\{\bx~\vert~\mathbf{E}\bx\leq \mathbf{f}\} \subseteq \R^{2N_t}$. For any extreme point  $\hat{\bx}$ of the set $\mathcal{F}$, there exists a subsystem $\mathbf{E}'\bx\leq \mathbf{f}'$, such that $\textnormal{rank}(\mathbf{E}')=2N_t$ and $\mathbf{E}'\hat{\bx}= \mathbf{f}'$.
\end{lemma}

\begin{proof}[Proof of Lemma \ref{lemma:extreme point}] 

For any extreme point $\hat{\bx}$ of the set $\mathcal{F}$, let $\mathbf{E}'\bx\leq \mathbf{f}'$ be the maximal subsystem of $\mathbf{E}\bx\leq \mathbf{f}$ such that $\mathbf{E}'\hat{\bx}= \mathbf{f}'$. Denote the remaining subsystem by $\mathbf{E}''\bx\leq \mathbf{f}''$, we have $\mathbf{E}''\hat{\bx}< \mathbf{f}''$. Assume by contradiction that $\text{rank}(\mathbf{E}')<2N_t$, then $\text{null}(\mathbf{E}') \neq \varnothing$. For any $\bx'\in \text{null}(\mathbf{E}')$, there exists a sufficiently small $\varepsilon$, such that 
$$
	\mathbf{E}(\hat{\bx}-\epsilon \bx')\leq \mathbf{f}, \quad \mathbf{E}(\hat{\bx}+\epsilon \bx')\leq \mathbf{f},
$$ 
i.e., $\bx_1:=\hat{\bx}-\epsilon \bx'\in \mathcal{F},~\text{and}~\bx_2 := \hat{\bx}+\epsilon \bx' \in \mathcal{F}$. However, $\hat{\bx} = (\bx_1+\bx_2)/2$ contradicts with the fact that $\hat{\bx}$ is an extreme point of $\mathcal{F}$.
\end{proof}
\noindent The proof of Theorem \ref{theorem: Ix lower bound} is as follows.
\begin{proof}
The feasible set of problem \eqref{eq:penalty problem} can be expressed as $\mathcal{F}=\{\bx~\vert~ \mathbf{E}\bx\leq \mathbf{f}\}$, where 
	\begin{equation}
		\mathbf{E}=\begin{bmatrix}
			-\bC \\
			\bD \\
			-\mathbf{I}\\
			\mathbf{I}
		\end{bmatrix}, \quad \mathbf{f}=\begin{bmatrix}
			-\boldsymbol{\gamma} \\
			\boldsymbol{\zeta} \\
			-\boldsymbol{\eta} \\
			\boldsymbol{\eta}
		\end{bmatrix}, \quad \boldsymbol{\eta}=\eta\cdot \mathbf{1}.
	\end{equation}
	As $\lambda>L_f/2$, $f(\bx)-\lambda\|\bx\|^2$ is strictly concave, then for any local minimizer $\hat{\bx}$ of problem \eqref{eq:penalty problem}, $\hat{\bx}$ must be an extreme point of $\mathcal{F}$.
	
	According to Lemma~\ref{lemma:extreme point}, there exists a subsystem $\mathbf{E}'\bx\leq \mathbf{f}'$, such that $\textnormal{rank}(\mathbf{E}')=2N_t$ and $\mathbf{E}'\hat{\bx}= \mathbf{f}'$. As rank$(\bC)\leq 2K < 2N_t$, rank$(\bD)=2$, at least $2N_t-\text{rank}(\bC)-\text{rank}(\bD)\geq 2N_t-2K-2$ inequalities in the subsystem $-\boldsymbol{\eta} \leq \bx \leq \boldsymbol{\eta}$ should hold with equality at $\hat{\bx}$. Hence, at least $2N_t-2K-2$ elements of $\hat{\bx}$ should be in $\{-\eta,\eta\}$, i.e., $I(\hat{\bx})\leq 2K+2.$ The proof is completed.
\end{proof}

\bibliographystyle{IEEEtran}
\bibliography{IEEEabrv,reference}

\begin{thebibliography}{10}
\providecommand{\url}[1]{#1}
\csname url@samestyle\endcsname
\providecommand{\newblock}{\relax}
\providecommand{\bibinfo}[2]{#2}
\providecommand{\BIBentrySTDinterwordspacing}{\spaceskip=0pt\relax}
\providecommand{\BIBentryALTinterwordstretchfactor}{4}
\providecommand{\BIBentryALTinterwordspacing}{\spaceskip=\fontdimen2\font plus
\BIBentryALTinterwordstretchfactor\fontdimen3\font minus \fontdimen4\font\relax}
\providecommand{\BIBforeignlanguage}[2]{{%
\expandafter\ifx\csname l@#1\endcsname\relax
\typeout{** WARNING: IEEEtran.bst: No hyphenation pattern has been}%
\typeout{** loaded for the language `#1'. Using the pattern for}%
\typeout{** the default language instead.}%
\else
\language=\csname l@#1\endcsname
\fi
#2}}
\providecommand{\BIBdecl}{\relax}
\BIBdecl

\bibitem{liu2018toward}
F.~Liu, L.~Zhou, C.~Masouros, A.~Li, W.~Luo, and A.~Petropulu, ``Toward dual-functional radar-communication systems: Optimal waveform design,'' \emph{{IEEE} Trans. Signal Process.}, vol.~66, no.~16, pp. 4264--4279, 2018.

\bibitem{hassanien2019dual}
A.~Hassanien, M.~G. Amin, E.~Aboutanios, and B.~Himed, ``Dual-function radar communication systems: A solution to the spectrum congestion problem,'' \emph{{IEEE} Signal Process. Mag.}, vol.~36, no.~5, pp. 115--126, 2019.

\bibitem{liu2022integrated}
F.~Liu, Y.~Cui, C.~Masouros, J.~Xu, T.~X. Han, Y.~C. Eldar, and S.~Buzzi, ``Integrated sensing and communications: Toward dual-functional wireless networks for {6G} and beyond,'' \emph{{IEEE} J. Sel. Areas Commun.}, vol.~40, no.~6, pp. 1728--1767, 2022.

\bibitem{walden1999analog}
R.~Walden, ``Analog-to-digital converter survey and analysis,'' \emph{{IEEE} J. Sel. Areas Commun.}, vol.~17, no.~4, pp. 539--550, 1999.

\bibitem{allen2011cmos}
P.~E. Allen and D.~R. Holberg, \emph{CMOS Analog Circuit Design}.\hskip 1em plus 0.5em minus 0.4em\relax Elsevier, 2011.

\bibitem{liu2024survey}
Y.-F. Liu, T.-H. Chang, M.~Hong, Z.~Wu, A.~Man-Cho~So, E.~A. Jorswieck, and W.~Yu, ``A survey of recent advances in optimization methods for wireless communications,'' vol.~42, no.~11, pp. 2992--3031, 2024.

\bibitem{saxena2017analysis}
A.~K. Saxena, I.~Fijalkow, and A.~L. Swindlehurst, ``Analysis of one-bit quantized precoding for the multiuser massive {MIMO} downlink,'' \emph{{IEEE} Trans. Signal Process.}, vol.~65, no.~17, pp. 4624--4634, 2017.

\bibitem{sohrabi2018one}
F.~Sohrabi, Y.-F. Liu, and W.~Yu, ``One-bit precoding and constellation range design for massive {MIMO} with {QAM} signaling,'' \emph{{IEEE} J. Sel. Topics Signal Process.}, vol.~12, no.~3, pp. 557--570, 2018.

\bibitem{wang2018finite}
C.-J. Wang, C.-K. Wen, S.~Jin, and S.-H. Tsai, ``Finite-alphabet precoding for massive {MU-MIMO} with low-resolution {DACs},'' \emph{{IEEE} Trans. Wireless Commun.}, vol.~17, no.~7, pp. 4706--4720, 2018.

\bibitem{shao2019framework}
M.~Shao, Q.~Li, W.-K. Ma, and A.~M.-C. So, ``A framework for one-bit and constant-envelope precoding over multiuser massive {MISO} channels,'' \emph{{IEEE} Trans. Signal Process.}, vol.~67, no.~20, pp. 5309--5324, 2019.

\bibitem{wu2024efficient2}
Z.~Wu, B.~Jiang, Y.-F. Liu, M.~Shao, and Y.-H. Dai, ``Efficient {CI}-based one-bit precoding for multiuser downlink massive {MIMO} systems with {PSK} modulation,'' \emph{{IEEE} Trans. Wireless Commun.}, vol.~23, no.~5, pp. 4861--4875, 2024.

\bibitem{choi2016near}
J.~Choi, J.~Mo, and R.~W. Heath, ``Near maximum-likelihood detector and channel estimator for uplink multiuser massive {MIMO} systems with one-bit {ADCs},'' \emph{{IEEE} Trans. Commun.}, vol.~64, no.~5, pp. 2005--2018, 2016.

\bibitem{mollen2017uplink}
C.~Moll\'en, J.~Choi, E.~G. Larsson, and R.~W. Heath, ``Uplink performance of wideband massive {MIMO} with one-bit {ADCs},'' \emph{{IEEE} Trans. Wireless Commun.}, vol.~16, no.~1, pp. 87--100, 2017.

\bibitem{shao2024accelerated}
M.~Shao, W.-K. Ma, J.~Liu, and Z.~Huang, ``Accelerated and deep expectation maximization for one-bit {MIMO-OFDM} detection,'' \emph{{IEEE} Trans. Signal Process.}, vol.~72, pp. 1094--1113, 2024.

\bibitem{ni2023uplink}
Z.~Ni, J.~A. Zhang, K.~Wu, and R.~P. Liu, ``Uplink sensing using {CSI} ratio in perceptive mobile networks,'' \emph{{IEEE} Trans. Signal Process.}, vol.~71, pp. 2699--2712, 2023.

\bibitem{keskin2021mimo}
M.~F. Keskin, H.~Wymeersch, and V.~Koivunen, ``{MIMO-OFDM} joint radar-communications: Is {ICI} friend or foe?'' \emph{{IEEE} J. Sel. Topics Signal Process.}, vol.~15, no.~6, pp. 1393--1408, 2021.

\bibitem{nowak2016co}
M.~J. Nowak, Z.~Zhang, Y.~Qu, D.~A. Dessources, M.~Wicks, and Z.~Wu, ``Co-designed radar-communication using linear frequency modulation waveform,'' in \emph{IEEE Mil. Commun. Conf. (MILCOM)}, 2016, pp. 918--923.

\bibitem{roberton2003}
M.~Roberton and E.~Brown, ``Integrated radar and communications based on chirped spread-spectrum techniques,'' in \emph{IEEE MTT-S Int. Microw. Symp. Digest, 2003}, vol.~1, 2003, pp. 611--614 vol.1.

\bibitem{saddik2007ultra}
G.~N. Saddik, R.~S. Singh, and E.~R. Brown, ``Ultra-wideband multifunctional communications/radar system,'' \emph{{IEEE} Trans. Microw. Theory Techn.}, vol.~55, no.~7, pp. 1431--1437, 2007.

\bibitem{liu2021cramer}
F.~Liu, Y.-F. Liu, A.~Li, C.~Masouros, and Y.~C. Eldar, ``Cram{\'e}r--{Rao} bound optimization for joint radar-communication beamforming,'' \emph{{IEEE} Trans. Signal Process.}, vol.~70, pp. 240--253, 2021.

\bibitem{liu2022transmit}
X.~Liu, T.~Huang, and Y.~Liu, ``Transmit design for joint {MIMO} radar and multiuser communications with transmit covariance constraint,'' \emph{{IEEE} J. Sel. Areas Commun.}, vol.~40, no.~6, pp. 1932--1950, 2022.

\bibitem{wen2023efficient}
C.~Wen, Y.~Huang, and T.~N. Davidson, ``Efficient transceiver design for {MIMO} dual-function radar-communication systems,'' \emph{{IEEE} Trans. Signal Process.}, vol.~71, pp. 1786--1801, 2023.

\bibitem{wu2024efficient}
J.~Wu, Z.~Wang, Y.-F. Liu, and F.~Liu, ``Efficient global algorithms for transmit beamforming design in {ISAC} systems,'' \emph{{IEEE} Trans. Signal Process.}, vol.~72, pp. 4493--4508, 2024.

\bibitem{yu2022precoding}
X.~Yu, Q.~Yang, Z.~Xiao, H.~Chen, V.~Havyarimana, and Z.~Han, ``A precoding approach for dual-functional radar-communication system with one-bit {DACs},'' \emph{{IEEE} J. Sel. Areas Commun.}, vol.~40, no.~6, pp. 1965--1977, 2022.

\bibitem{tsinos2021joint}
C.~G. Tsinos, A.~Arora, S.~Chatzinotas, and B.~Ottersten, ``Joint transmit waveform and receive filter design for dual-function radar-communication systems,'' \emph{{IEEE} J. Sel. Topics Signal Process.}, vol.~15, no.~6, pp. 1378--1392, 2021.

\bibitem{wu2025quantized}
Z.~Wu, Y.-F. Liu, W.-K. Chen, and C.~Masouros, ``Quantized constant-envelope waveform design for massive {MIMO} {DFRC} systems,'' \emph{{IEEE} J. Sel. Areas Commun.}, vol.~43, no.~4, pp. 1056--1073, 2025.

\bibitem{wang2025interference}
Y.~Wang, X.~Hu, A.~Li, C.~Masouros, K.-K. Wong, and K.~Yang, ``Interference exploitation in {ISAC} systems: Finite-alphabet precoding with low resolution {DACs} and {PSs},'' \emph{{IEEE} Trans. Wireless Commun.}, vol.~25, pp. 2264--2279, 2026.

\bibitem{liu2021dual}
R.~Liu, M.~Li, Q.~Liu, and A.~L. Swindlehurst, ``Dual-functional radar-communication waveform design: A symbol-level precoding approach,'' \emph{{IEEE} J. Sel. Topics Signal Process.}, vol.~15, no.~6, pp. 1316--1331, 2021.

\bibitem{an2023fundamental}
J.~An, H.~Li, D.~W.~K. Ng, and C.~Yuen, ``Fundamental detection probability vs. achievable rate tradeoff in integrated sensing and communication systems,'' \emph{{IEEE} Trans. Wireless Commun.}, vol.~22, no.~12, pp. 9835--9853, 2023.

\bibitem{xu2024mimo}
C.~Xu and S.~Zhang, ``{MIMO} integrated sensing and communication exploiting prior information,'' \emph{{IEEE} J. Sel. Areas Commun.}, vol.~42, no.~9, pp. 2306--2321, 2024.

\bibitem{cheng2021transmit}
Z.~Cheng, S.~Shi, Z.~He, and B.~Liao, ``Transmit sequence design for dual-function radar-communication system with one-bit {DACs},'' \emph{{IEEE} Trans. Wireless Commun.}, vol.~20, no.~9, pp. 5846--5860, 2021.

\bibitem{wang2023joint}
B.~Wang, H.~Li, and Z.~Cheng, ``Joint transceiver design for massive {MIMO DFRC} systems with one-bit {DACs/ADCs},'' in \emph{IEEE Globecom Workshops (GC Wkshps)}, 2023, pp. 649--654.

\bibitem{sun2026one}
Y.~Sun, R.~Liu, M.~Li, and Q.~Liu, ``One-bit {DAC/ADC} transceiver designs for efficient {MIMO}-{ISAC} systems,'' \emph{{IEEE} Trans. Commun.}, vol.~74, pp. 4694--4709, 2026.

\bibitem{masouros2015exploiting}
C.~Masouros and G.~Zheng, ``Exploiting known interference as green signal power for downlink beamforming optimization,'' \emph{{IEEE} Trans. Signal Process.}, vol.~63, no.~14, pp. 3628--3640, 2015.

\bibitem{wu2023diversity}
Z.~Wu, J.~Wu, W.-K. Chen, and Y.-F. Liu, ``Diversity order analysis for quantized constant envelope transmission,'' \emph{{IEEE} Open J. Signal Process.}, vol.~4, pp. 21--30, 2023.

\bibitem{skolnik2002introduction}
M.~I. Skolnik, \emph{Introduction to Radar Systems}.\hskip 1em plus 0.5em minus 0.4em\relax McGraw-Hill Education, 2002.

\bibitem{deng2022receive}
M.~Deng, Z.~Cheng, and Z.~He, ``Receive filter design for {MIMO} radar with one-bit {ADCs},'' \emph{Digit. Signal Process.}, vol. 123, p. 103363, 2022.

\bibitem{xi2020gridless}
F.~Xi, Y.~Xiang, S.~Chen, and A.~Nehorai, ``Gridless parameter estimation for one-bit {MIMO} radar with time-varying thresholds,'' \emph{{IEEE} Trans. Signal Process.}, vol.~68, pp. 1048--1063, 2020.

\bibitem{stoica2021cramer}
P.~Stoica, X.~Shang, and Y.~Cheng, ``The {Cram\'er--Rao} bound for signal parameter estimation from quantized data [lecture notes],'' \emph{{IEEE} Signal Process. Mag.}, vol.~39, no.~1, pp. 118--125, 2021.

\bibitem{host2000effects}
A.~Host-Madsen and P.~Handel, ``Effects of sampling and quantization on single-tone frequency estimation,'' \emph{{IEEE} Trans. Signal Process.}, vol.~48, no.~3, pp. 650--662, 2000.

\bibitem{liu2024extreme}
J.~Liu, Y.~Liu, W.-K. Ma, M.~Shao, and A.~M.-C. So, ``Extreme point pursuit --- part {I}: A framework for constant modulus optimization,'' \emph{{IEEE} Trans. Signal Process.}, vol.~72, pp. 4541--4556, 2024.

\bibitem{andreani2007on}
R.~Andreani, E.~G. Birgin, J.~M. Mart\'{\i}nez, and M.~L. Schuverdt, ``On augmented {Lagrangian} methods with general lower-level constraints,'' \emph{SIAM J. Optim.}, vol.~18, no.~4, pp. 1286--1309, 2008.

\bibitem{calamai1987projected}
P.~H. Calamai and J.~J. Mor{\'e}, ``Projected gradient methods for linearly constrained problems,'' \emph{Math. Program.}, vol.~39, no.~1, pp. 93--116, 1987.

\bibitem{hager2006new}
W.~W. Hager and H.~Zhang, ``A new active set algorithm for box constrained optimization,'' \emph{SIAM J. Optim.}, vol.~17, no.~2, pp. 526--557, 2006.

\bibitem{friedlander1994new}
A.~Friedlander, J.~Mart{\'\i}nez, and S.~Santos, ``A new trust region algorithm for bound constrained minimization,'' \emph{Appl. Math. Optim.}, vol.~30, no.~3, pp. 235--266, 1994.

\bibitem{barzilai1988two}
J.~Barzilai and J.~M. Borwein, ``Two-point step size gradient methods,'' \emph{IMA J. Numer. Anal.}, vol.~8, no.~1, pp. 141--148, 1988.

\bibitem{grippo1986nonmonotone}
L.~Grippo, F.~Lampariello, and S.~Lucidi, ``A nonmonotone line search technique for {Newton}'s method,'' \emph{SIAM J. Numer. Anal.}, vol.~23, no.~4, pp. 707--716, 1986.

\bibitem{galvan2020alternating}
G.~Galvan, M.~Lapucci, T.~Levato, and M.~Sciandrone, ``An alternating augmented {Lagrangian} method for constrained nonconvex optimization,'' \emph{Optim. Methods Softw.}, vol.~35, no.~3, pp. 502--520, 2020.

\end{thebibliography}

\end{document}